\documentclass[a4paper,11pt,fleqn]{article}
\usepackage[utf8]{inputenc}
\usepackage{dsfont}

\usepackage{caption} 

\usepackage{geometry}

\usepackage{amsmath}
\usepackage{amsfonts,amssymb,amsthm}
\usepackage{xcolor}
\definecolor{linkColor}{RGB}{0, 128, 128}
\definecolor{citeColor}{RGB}{0, 112, 64}
\definecolor{urlColor}{RGB}{120, 0, 120}
\usepackage{hyperref}
\hypersetup{
    colorlinks=true,
    linkcolor=linkColor,
    citecolor=citeColor,      
    urlcolor=urlColor,
}
\usepackage{url}

\definecolor{oracleColor}{HTML}{A0BDD3} 
\definecolor{noiseColor}{HTML}{F2B547} 
\definecolor{algoColor}{HTML}{E6C0BA} 
\definecolor{phiStateCol}{RGB}{0, 140, 25} 

\theoremstyle{plain}
\newtheorem{thm}{Theorem}
\newtheorem{lem}[thm]{Lemma}

\newtheorem{clm}[thm]{Claim}

\newtheorem{rmk}[thm]{Remark}

\usepackage{tikz}
\usetikzlibrary{angles,quotes} 
\usetikzlibrary{decorations.pathreplacing,calligraphy}

\usepackage{calc} 

\newcommand{\cA}{\mathcal{A}}

\newcommand{\cH}{\mathcal{H}}

\renewcommand{\>}{\rangle}
\newcommand{\<}{\langle}

\DeclareMathOperator{\tr}{\mathrm{Tr}}

\newcommand{\qAnd}{\quad\text{and}\quad}

\newcommand{\Func}{{\mathcal{F}}} 

\usepackage{mathrsfs} 
\newcommand{\ccD}{\mathscr{D}}
\newcommand{\ccI}{\mathscr{I}}
\newcommand{\ccN}{\mathscr{N}}

\newcommand{\ccO}{\mathscr{O}}

\newcommand{\regA}{\mathsf{A}}

\newcommand{\regQ}{\mathsf{Q}}
\newcommand{\regQi}{\mathsf{Q_i}}
\newcommand{\regQo}{\mathsf{Q_o}}
\newcommand{\regR}{\mathsf{R}}
\newcommand{\regT}{\mathsf{T}}
\newcommand{\regW}{\mathsf{W}}

\newcommand{\regFunc}{\regT} 

\newcommand\inrec[1]{{#1_{\mathrm{in}}}}

\newcommand\unif{u} 
\newcommand{\unifR}{v} 

\newcommand\actv{{act}} 
\newcommand\pasv{{pas}} 

\newcommand\marked{{\pmb{0}}}

\newcommand{\OO}{\mathrm{O}}

\newcommand\progA{{\mathfrak{A}}} 
\newcommand\progB{{\mathfrak{B}}} 
\newcommand\progC{{\mathfrak{C}}} 

\newcommand\flagC{{\mathfrak{c}}}
\newcommand\flagQ{{\mathfrak{q}}}

\newcommand\app{\wr}
\newcommand\rec{{\mathfrak{R}}}

\title{Quantum Search with Noisy Oracle}
\author{
Ansis Rosmanis\thanks{E-mail: \texttt{rosmanis@math.nagoya-u.ac.jp}}\\ [.5ex]
\normalsize  Graduate School of Mathematics \\ 
\normalsize Nagoya University
}
\date{September 26, 2023}

\begin{document}

\maketitle

\begin{abstract}
We consider quantum search algorithms that have access to a noisy oracle that, for every oracle call, with probability $p>0$ completely depolarizes the query registers, while otherwise working properly. Previous results had not ruled out quantum $\OO(\sqrt{n})$-query algorithms in this setting, even for constant $p$. We show that, for all $p\le 0.99$, the quantum noisy-query complexity of the unstructured search is $\widetilde\Theta(\max\{np,\sqrt{n}\})$. The lower bound $\Omega(\max\{np,\sqrt n\})$ holds also for the dephasing noise and even when, for every oracle call, the algorithm is provided with a flag indicating whether the error has occurred.
\end{abstract}

\section{Introduction}

Unstructured search is one of the most studied computational tasks in quantum computing due to its applicability to many other problems and due to Grover's algorithm performing this task quadratically faster than any classical algorithm \cite{grover:search}. Even before Grover presented his algorithm, Bennett, Bernstein, Brassard, and Vazirani showed that quadratic quantum speedups is the best one can hope for \cite{bennett:searchLowerB}. More formally, these works showed that finding a marked element among $n$ elements has the bounded-error quantum query complexity $\Theta(\sqrt{n})$.


Since then, many subsequent works have shed more light on the problem. Zalka showed that quadratic speedups do not survive parallelization~\cite{zalka:Grover}. In the same work, it was shown that Grover's algorithm is the very best algorithm for the problem and that not even constant additive gains can be made in quantum query complexity over it (see also \cite{dohotaru:GroverOpt} for an alternative proof). Recent works have also investigated hybrid quantum-classical algorithms that can access the oracle both in a quantum, coherent manner as well as in a classical, non-coherent manner, and shown that a sublinear number of classical queries cannot substantially speed up the quantum search \cite{Rosmanis:2022:hybrid,hamoudi:2022:tradeoffs}. Yet, multiple important questions regarding unstructured search remain open.
In this work we focus on search with faulty oracles. 

\subsection{Previous Work on Search with Faulty Devices}

Due to the paramount importance of unstructured search, many works have studied the problem in various settings where the computation is affected by certain faults. These studies can be categorized along various lines. For example, they may differ in
whether only Grover's algorithm is addressed or any potential algorithm for the search,
whether the faults only affect the oracle or also other (i.e., input-independent) components of the algorithm, whether the faulty oracle operations are a result of some subroutine and, thus, are reversible or are the faults irreversible, and in many other aspects.

Here, let us first focus on studies that address Grover's algorithm and its variations specifically, and then on studies that place limits on any general search algorithm.

\paragraph{Grover's algorithm with faults.}

When the noise is reversible, the speedups achieved by Grover's algorithm survive, and, unlike in the classical case, does not even suffer from slowdowns by logarithmic factors~\cite{Hoyer+M+dW-BoundedErrorSearch-03,Suzuki+Y+N+W-RobustSearch-06}. Similar results also hold for problems other than the unstructured search \cite{Buhrman+N+R+dW-RobustPolys-07,Iwama+R+Y-BiasedOracles-2006}.

When the noise is not unitary, however, Grover's algorithm typically suffers in performance. Depending on the noise model, one either has to run the algorithm for longer time (i.e., more iterations), or its maximum success probability decreases. In one of the earliest studies of Grover's algorithm with noise \cite{PN+RA_GroverNoise_99}, Pablo-Norman and Ruiz-Altaba studied noise (claimed physically unrealistic by the authors) that before every oracle call perturbs the vector corresponding to the pure state of the computation by first adding to it a random vector whose norm is chosen according to the half-normal distribution of standard deviation $\sigma$ and then renormalizing the resulting vector. They showed that, for such a perturbation, the success probability of Grover's algorithm vanishes unless $\sigma$ is of order $\OO(n^{-2/3})$.

A more commonly studied noise model is that where the oracle always works properly for unmarked elements, but it might sometimes falsely report marked elements as unmarked. In particular, while the standard noiseless phase oracle used in Grover's algorithm maps $|x\>$ to $|x\>$ for unmarked $x$ and $|x\>$ to $-|x\>$ for marked $x$, in the noisy setting, for every marked element $x$, the state $|x\>$ is instead mapped to $-\mathrm{e}^{\pi\mathrm{i}\varepsilon_x}|x\>$ for some phase error $\varepsilon_x$ chosen randomly.

When there is only a single marked element $x$, the problem has been investigated both numerically~\cite{Long+++Grover00} and analytically~\cite{Shenvi++NoisyOracleSearch03}. 
Shenvi, Brown, and Whaley~\cite{Shenvi++NoisyOracleSearch03} showed that, for $p:=\sqrt{\mathbb{E}[\varepsilon_x^2]}$, the noise effectively ruins the computation after $1/p^2$ iterations of Grover's algorithm, at which point the success probability of $\Omega(\min\{1/(np^4),1\})$ can be achieved. Hence, if $p=\OO(n^{1/4})$, a single run of Grover's algorithm is sufficient. On the other hand, if $p=\omega(n^{1/4})$, the above process can be repeated $np^4$ times, resulting in the total number of oracle calls being $\OO(np^2)$, which still surpasses the classical computation for subconstant $p$.
Shenvi et al.~also showed similar results for the continuous-time version of Grover's algorithm, introduced in~\cite{Fahri+Gutmann_analogGrover_1998}.

When there are multiple marked elements $x$, the spacial case where $\varepsilon_x\in\{0,1\}$ has been studied~\cite{Ambainis+++GroverWithErrors13,Kravchenko+N+R_SomeFaulty_2018}.
Note that $\varepsilon_x=1$ corresponds to falsely reporting $x$ unmarked, while $\varepsilon_x=0$ corresponds to a faultless oracle call.
Kravchenko, Nahimovs, and Rivosh~\cite{Kravchenko+N+R_SomeFaulty_2018} showed that, if faults occur only for one marked element, then quadratic quantum speedups still persist. On the other hand, if for every marked element faults occur with at least a constant probability, then after $\OO(n)$ queries the state converges to a certain mixed state from which a marked element can be obtained with probability close to one \cite{Ambainis+++GroverWithErrors13}.

The performance of Grover's algorithm has been also studied when the noise occurs not in the oracle, but in the Hadamard gates used to implement the so-called Grover diffusion operator.
Shapira, Mozes, and Biham~\cite{Shapira++GroverUnitaryNoise03}
considered the scenario where each Hadamard gate $H$ is replaced by a unitary $\widetilde H$ that is a slight perturbation of $H$. They show that if $\widetilde H$ and $H$ are close, meaning the perturbation is below a certain threshold, then the bitstring returned by the algorithm is either the marked element or differs from the marked element only on few bits. On the other hand, if the perturbation is above the threshold, the search completely fails.

\paragraph{General search algorithm.}

While it is important to study how the noise affects Grover's algorithm, showing that Grover's algorithm succumbs under certain noise does not rule out a possibility that some other algorithm might resist it.

Regev and Schiff \cite{regev:faultySearch} considered faulty oracle very similar to that considered in~\cite{Long+++Grover00,Shenvi++NoisyOracleSearch03,Ambainis+++GroverWithErrors13,Kravchenko+N+R_SomeFaulty_2018}.
Namely, they considered a scenario where the oracle is negligent and, independently for each oracle call,  with some constant probability $p$ it forgets to apply the unitary query operator $O_f$, instead simply applying the identity operator $I$. Regev and Schiff showed that, in such a scenario, for any quantum algorithm, $\Omega(n)$ queries are required to perform search.
Later, Temme~\cite{Temme-FaultyOracle-2014} showed an analogous result in the continuous-time quantum query model.

Both results rely on the fact that, when there are no marked elements, the oracle works faultlessly, and the computation remains pure and can be used as a ``measuring rod" against which the progress of the computation is measured. The computation remaining pure is something that is not true for the depolarizing noise, a common noise model considered in quantum computing.

The depolarizing noise, acting on a given register, with probability $p$ maps any state of that register to the maximally mixed state, effectively erasing the register, while with probability $1-p$ it does nothing.
Vrana, Reeb, Reitzner, and Wolf~\cite{Vrana+R+R+W_FaultIgnorant_2014} studied the setting where, between any two consecutive oracle calls, the depolarizing noise occurs on the whole memory of the algorithm. They showed that, in such a scenario, the search requires at least $\Omega(np)$ queries.

In a recent work \cite{ChenCHL:2022:nisq}, Chen, Cotler, Huang, and Li formulated the computational class of noisy intermediate-scale quantum (NISQ) computation, where every qubit individually and independently is affected by the depolarizing noise of rate $p$. Among other results, they showed that in such a scenario a noisy quantum algorithm cannot perform quantum search faster than in time $\tilde\Theta(np)$, also providing a matching upper bound.

In contrast to \cite{Vrana+R+R+W_FaultIgnorant_2014} and~\cite{ChenCHL:2022:nisq}, in this work we consider the scenario where the depolarizing noise only affects the query registers, which are used to access the oracle, and our hardness results will also hold for weaker noise models, such as the dephasing noise, which with probability $p$ projects the query registers to the computational basis.
While our hardness results encompass those in \cite{Vrana+R+R+W_FaultIgnorant_2014}, they are not directly comparable to the ones in \cite{ChenCHL:2022:nisq} due to different noise models.

\subsection{Results}

In this paper, we consider quantum query algorithms where the input-independent operations of the algorithm are faultless, but the interaction with the oracle can suffer from certain types of noise.
We illustrate such a noisy interaction in Figure~\ref{fig:informal_interact}.

\begin{figure}[!h]
\centering
\begin{tikzpicture}
\draw[thick,->] (-1.5,-1) -- (-1.25,-1.5);
\draw[thick, <-] (.9,.5) arc
    [
        start angle=50,
        end angle=130,
        x radius=1.4 cm,
        y radius =1.2cm
    ] ;
\draw[thick, ->] (.9,-.5) arc
    [
        start angle=-50,
        end angle=-130,
        x radius=1.4 cm,
        y radius =1.2cm
    ] ;    
\draw[thick, fill=algoColor] (-2,0) circle (1);
\draw[thick, fill=oracleColor] (2,0) circle (1);
\node [black] at (-2,0) {Algorithm};
\node [above, black] at (2,0) {Oracle};
\node [below, black] at (2,0) {for $f$};
\node [below, black] at (-1.25,-1.5) {Result};
 \draw[noiseColor, fill=noiseColor!30!white, fill opacity = 0.75]
  (0.8,0) \foreach \i in {1,...,23}{ -- ({(0.6+0.2*cos(180*\i))*cos(15*\i)},{(0.9+0.3*cos(180*\i))*sin(15*\i)}) } -- cycle;
\node [below, noiseColor!60!brown] at (0,-1.1) {Noise};
\end{tikzpicture}
\captionsetup{font=small}
\captionsetup{width=0.9\textwidth}
\caption[my caption]{%
An informal illustration of a noisy-query algorithm. The algorithm, initially knowing nothing about the problem instance $f$, is given access to it via an oracle. Each such interaction, however, is faulty with probability $p$. At the end of the computation, the algorithm has to return the correct answer for $f$.}
\label{fig:informal_interact}
\end{figure}
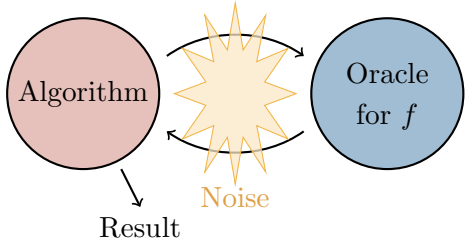

The strongest noise model that we consider is that of the depolarizing noise. In such a noise model, independently for each oracle call, with probability $1-p$ the oracle works properly, while with probability $p$ it replaces the state of all the query registers with the maximally mixed state, effectively erasing those registers.
While this noise model is seemingly much harsher than the negligent oracle considered by Regev and Schiff~\cite{regev:faultySearch} and Temme~\cite{Temme-FaultyOracle-2014}, for the depolarizing noise, quadratic speedups by quantum algorithms were not ruled out even for constant $p$.

In this work we show that the depolarizing oraclar noise indeed thwarts quantum speedups.
We also show the same hardness result for the dephasing noise that, independently for every oracle call, in addition to applying the proper oracle operation, with probability $p$ projects the query registers to the computational basis.

\begin{thm}
\label{thm:main}
For both the depolarizing and the dephasing noise of rate $p>0$, the $\epsilon$-error quantum noisy-query complexity of the unstructured search is at least $np(1-\epsilon)/49 - 1$. 
This lower bound holds even if for every oracle call the algorithm is provided with a bit indicating whether the oracle call was performed properly.
\end{thm}

When $p\ge 1/\sqrt{n}$, this bound can be seen to be tight by, informally speaking, $\OO(np^2)$ times running Grover's algorithm for $\OO(1/p)$ steps and then checking the correctness of each solution. For the dephasing noise, the correctness of the solution can be checked using a single query, while, for the depolarizing noise (assuming $1-p=\Omega(1)$), a logarithmic number of queries might be necessary to obtain the desired confidence.
For $p=o(1/\sqrt{n})$, one expects to be able to run the entire Grover's algorithm before the first faulty oracle call occurs.
Hence, $\OO(\max\{np(1+p\log n),\sqrt n\})$ noisy queries suffice, where the $p\log n$ term can be omitted for the dephasing noise.

As stated in Theorem~\ref{thm:main}, we show that the extra bit of information indicating whether the oracle call was performed properly does not help the algorithm to overcome effects of the depolarizing and the dephasing noise.
 Let us note that, in the case of the negligent oracle considered by Regev and Schiff~\cite{regev:faultySearch}, having such a bit would restore the quadratic speedups achieved by Grover's algorithm.
 
\subsection{Techniques}

In Theorem~\ref{thm:main}, the flag bit that indicates whether the error has occurred, in addition to strengthening the lower bound, plays another role in our proof of the theorem. In particular, because of it, it is sufficient to prove the result only for the dephasing noise. Indeed, if there were faster algorithms for the depolarizing noise, those algorithms could be transformed into algorithms for the dephasing noise of the same query complexity. In particular, whenever the flag bit received by the algorithm indicates that the dephasing has occurred, the algorithm can depolarize all query registers, in effect simulating the depolarizing noise.

The dephasing noise, in turn, has a close connection to the classical oracle calls considered in works on hybrid quantum-classical search algorithms in \cite{Rosmanis:2022:hybrid,hamoudi:2022:tradeoffs}. 
Indeed, the faultless quantum oracle call followed by a completely dephasing noise is exactly the same as the classical oracle call considered in those works. The difference now is that, unlike in the hybrid scenario, the algorithm has no control over when this effectively-classical oracle call will happen, every oracle call being classical with probability $p$.

For the noisy oracle scenario considered in this work, proof techniques by Hamoudi, Liu, and Sinha \cite{hamoudi:2022:tradeoffs}, which allow one to deal with mixed-state computation, are more useful than the ones in \cite{Rosmanis:2022:hybrid}, which stay closer to the approach in \cite{regev:faultySearch}. In addition to being inspired by techniques in \cite{hamoudi:2022:tradeoffs}, we also draw inspiration from the quantum lower bound for search in \cite{Rosmanis2021}.
Both of these works, in turn, are inspired by Zhandry's compressed oracle approach \cite{zhandry:record}. While \cite{hamoudi:2022:tradeoffs} builds upon \cite{zhandry:record}  by showing how to simultaneously handle both classical and quantum queries, \cite{Rosmanis2021} shows how to avoid compression-decompression steps and thus allows handling scenarios where the output of the oracle on one input may depend on that on another input (as is the case when searching for a unique marked element).

Similarly to the adversary bound~\cite{ambainis:adversary} and many of its variants, we introduce the truth-table register, which contains the full description of the function $f$ computed by the oracle. This register then controls actions of the oracle $O_f$. 
Furthermore, similarly to \cite{hamoudi:2022:tradeoffs}, we also introduce what we call the ``record'' registers that purify the overall system. We then use the joint state of the truth-table and record registers, which are the registers not directly accessible by the algorithm, to define a certain metric that measures the progress of the computation. More precisely, similarly to \cite{hamoudi:2022:tradeoffs}, we decompose the entire memory space (including that of the algorithm) into three parts as $\cH^\progA\oplus\cH^\progB\oplus\cH^\progC$, where $\cH^\progC$ essentially corresponds to the scenario where the classical oracle has succeeded (that is, the dephasing noise has collapsed the query input register to a marked input $x$), $\cH^\progB$ corresponds to the space orthogonal to $\cH^\progC$ where the quantum oracle has found a marked input, and $\cH^\progA$ corresponds to the space where no substantial progress has been made. All three spaces are invariant under linear operations on the algorithm registers alone.
Initially, all the probability weight is on the subspace $\cH^\progA$, but, in order for the algorithm to succeed, this probability weight has to be transferred to $\cH^\progB\oplus\cH^\progC$. 

Departing from previous proof techniques, now we have to go further and decompose the space $\cH^\progB$ as $\cH^{\progB,\actv}\oplus \cH^{\progB,\pasv}$, where, this time, the ``active'' subspace $\cH^{\progB,\actv}$ and the ``passive'' subspace $\cH^{\progB,\pasv}$ are not invariant under linear operations on the algorithm registers. 
The active subspace $\cH^{\progB,\actv}$ is the subspace of $\cH^{\progB}$ to which the probability weight from $\cH^\cA$ can be partially transferred, and the subspace which would be used by a noiseless execution of Grover's algorithm. Unfortunately, however, $\cH^{\progB,\actv}$ is also the subspace affected by the noise. On the other hand, the passive subspace $\cH^{\progB,\pasv}$ can be used to shield the quantum memory from noise, yet this shielding thwarts the progress of the computation. So, we show that, in this case, you can't have your cake and eat it too; namely, we show that you cannot simultaneously both progress the computation and avert its corruption by noise. As a result, we show the tight bound $\Omega(np)$ as given in Theorem~\ref{thm:main}.

\subsection{Organization}

First, in Section~\ref{sec:model}, we introduce the computational model of the quantum noisy-query algorithm, as well as formalize the unstructured search problem. Then, Sections~\ref{sec:envRegs}--\ref{sec:boundProg} are devoted to proving our main result, the lower bound given by Theorem~\ref{thm:main}.

The non-unitary noise renders the state of the computation to be mixed. In Section~\ref{sec:envRegs}, we introduce two additional registers that purify the computation, one purifying the non-unitary action of the noise, while the other purifies the random choice among the hardest instances of the search problem, i.e., instances having a unique marked element. In Section~\ref{sec:progressMes} we use these two registers to define the progress measure of the computation, and quantify the total change this measure must undergo in order for the computation to succeed. Then, in Section~\ref{sec:transit}, we present various claims towards limiting by how much a single oracle call can change the progress measure, and, finally, we complete the proof in Section~\ref{sec:boundProg}.

In Section~\ref{sec:algoSketch}, we briefly sketch the algorithm that shows that our lower bound is almost asymptotically tight. We conclude with a brief discussion in Section~\ref{sec:discussion}.

\bigskip

The main body of the text concerns the unstructured search in the worst case setting, and it is self-contained, meaning that, none of the proofs are deferred to the appendix. Our proof techniques, however, also work when one is interested in the noisy search for random functions, and, in Appendix~\ref{app:randFunc}, we sketch how an equivalent bound to our main result can be obtained in that setting.

\section{Model of Computation}
\label{sec:model}

We assume that the reader is familiar with basic concepts of quantum computation. For introductory texts, see, for example~\cite{NielsenChuang,watrous2018}.

In this section we define query algorithms with noisy oracles for a rather general computational problem.
We then focus on the unstructured search problem in particular and describe its various forms at the end, in Section~\ref{sec:SearchDef}.

\subsection{Quantum memory}
\label{sec:QuantMemo}

The memory of a quantum algorithm is organized in registers. Each register is associated with some finite set $S$ and a complex Euclidean space of dimension $|S|$, denoted $\mathbb{C}^{S}$. The \emph{standard basis} or the \emph{computational basis} of this space is some fixed orthonormal basis whose vectors are uniquely labeled by the elements of $S$, and we write it as $\{|s\>\colon s\in S\}$. The pure states of the register are (column) unit vectors in $\mathbb{C}^{S}$. 

A qubit is a register associated with the set $\{0,1\}$, and multiple qubits can be grouped together into a larger register.
We will also occasionally associate a qubit with the set $\{\flagQ,\flagC\}$, where $\flagQ:=0$ and $\flagC:=1$ are flags with $\flagC$ indicating that the noise has caused the error to occur and thus the oracle call is classical in nature, and $\flagQ$ indicating no error and the oracle call being quantum in nature.

We may also write a pure state $|\psi\>$ as its corresponding density operator $|\psi\>\<\psi|$. The quantum memory can also be in a mixed-state, in which case its corresponding density operator is a convex combination over pure states $|\psi\>\<\psi|$.

We use capital letters $\regA,\regQ,\regR, \regFunc, \regW$ in sans serif font to denote registers, and we might add them as subscripts to operators to highlight on which registers these operators act. 

The memory of a quantum algorithm typically consists of multiple registers, and the state of the entire memory is a density operator on the tensor product of the Euclidean spaces corresponding to each register. The evolution of quantum memory is governed by completely-positive trace-preserving (CPTP) maps, also known as quantum channels, an evolution governed by unitary operators being a special case. When an operator or a CPTP map acts as the identity on some registers, we typically omit those registers from the notation. We use $I_k$ to denote the $k$-dimensional identity operator, and we may drop the subscript $k$ when it is clear from the context.

\subsection{Oracle calls}
\label{ssec:oracles}

Let us here introduce various forms of the quantum oracle and notation pertaining to it. Let $[n]:=\{0,1,2,\ldots,n-1\}$ here and throughout the text.

We consider algorithms that have oracle access to a function $f\colon[n]\rightarrow[m]$. We call such algorithms \emph{query algorithms}, and we refer to each oracle access to $f$ by using  terms an ``oracle call'' and a ``query'' interchangeably. Let us assume, without much loss of generality, that $m$ is a power of $2$, and we may interpret the elements of $[m]$ as $\log m$-bit strings. In the main text, we consider only the scenario $m=2$ and, thus, functions $f\colon[n]\rightarrow\{0,1\}$, while the $m=\Omega(n)$ case is relevant in the random function scenario considered in the appendix.

As we will discuss in more detail in Section~\ref{sec:NoisyAlgo}, the algorithm registers $\regA$ will consist of the query register $\regQ$ and the workspace register $\regW$. When considering oracles with error-indicating flag qubits, each oracle call will grow the workspace register by a qubit.

\def\ygap{0.13}

\begin{figure}[!h]
\centering
\begin{tikzpicture}
  \foreach \y in {0,...,5}
      \draw (-1.45,\ygap*\y)--(1.45,\ygap*\y) ;
   \draw (-1.45,-\ygap*4.5)--(1.45,-\ygap*4.5);
   \draw [draw=black,fill=oracleColor] (-.7,-0.9) rectangle (.7,0.9);
   \node at (0,0) {$O_f$};
   \node [right] at (1.7,\ygap*2.5) {$|x\>$};
   \node [right] at (1.7,-\ygap*4.5) {$|y\oplus f(x)\>$};
      \node [right] at (-2.4,\ygap*2.5) {$|x\>$};
    \node [right] at (-2.4,-\ygap*4.5) {$|y\>$};
\end{tikzpicture}
\captionsetup{font=small}
\captionsetup{width=0.9\textwidth}
\caption[my caption]{%
The circuit diagram representation of the noiseless standard oracle call. Here and in subsequent illustrations we only consider the binary case, when  $m=2$. Here $x\in[n]$ and $y\in\{0,1\}$. The oracle can be called in a superposition.}
\label{fig:noiseless_orac}
\end{figure}
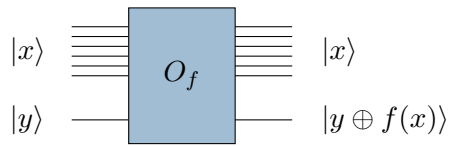

\paragraph{Noiseless oracle.}

The oracle call acts on the \emph{query} register $\regQ=\regQi\regQo$, which is composed of two subregisters: the \emph{query-input} register $\regQi$ corresponding to the set $[n]$ and the \emph{query-output} register $\regQo$ corresponding to $[m]$.
The query-input register and the query-output registers are also commonly referred to as, respectively, the index and the target register.
 The oracle call to $f$ is the unitary
\[
 O_f
 :=\sum_{x\in[n]}\sum_{y\in [m]} |x,y\oplus f(x)\>\<x,y|.
 \]
where we think of $f(x)$ and $y$ as $\log m$-bit strings and $f(x)\oplus y$ denotes their entry-wise exclusive OR (see Figure~\ref{fig:noiseless_orac} for the circuit diagram representation of $O_f$ in the binary case, $m=2$). 
The operator $O_f$ is sometimes referred to as the \emph{standard} oracle, as opposed to the \emph{phase} oracle 
\[
 O^{\mathrm{ph}}_f:= \sum_{x\in[n]}\sum_{y\in[m]} (-1)^{f(x)\cdot y}|x,y\>\<x,y|,
\]
where `$\cdot$' denotes the inner product of bit strings. 
 The two oracles are equivalent up to a basis change of the query-output register.%
\footnote{\label{fn:targetBasis}%
The basis choice for the query-output register could play a role when one talks about dephasing that qubit, as the dephasing noise is basis-specific. However, in this work we already show the hardness of the case when the dephasing noise affects only the query-input register. That, together with flag qubits indicating the presence of error, will also imply hardness for stronger noise models.}
In the binary case, $m=2$, we can write
\[
 O_f  =  I_{2n} - 2\sum_{x\in f^{-1}(1)}|x,-\>\<x,-|,
\]
where $|-\>:=(|0\>-|1\>)/\sqrt 2$ a Hadamard basis state.
We denote the CPTP map corresponding to $O_f$ by $\ccO_f$.
We may call $O_f$ the \emph{noiseless} or the \emph{faultless} oracle (call) to emphasize the absence of noise.

\paragraph{Noisy oracles.}

We will define noisy oracle calls as CPTP maps on the query register $\regQ$ as the composition of the noiseless oracle call $\ccO_f$ and a probabilistic, non-unitary noise $\ccN_p$.
For  $\rho$ a linear operator on $\regQ$,
we define the \emph{completely} depolarizing channel (of all the query registers $\regQ=\regQi\regQo$) as the CPTP map
\[
\ccN_1^{polar}\colon \rho \mapsto \tr[\rho] \,I_\regQ/(nm)
\]
and the completely dephasing channel (of the query-input register $\regQi$) as the CPTP map
\[
\ccN_1^{phase}\colon \rho \mapsto \sum_{x\in[n]} \big(|x\>\<x|\otimes I_\regQo\big)\rho \big(|x\>\<x|\otimes I_\regQo\big),
\]
where $\tr[\,\cdot\,]$ denotes the trace.

Note that we have chosen to define the completely dephasing channel so as not to affect the query-output register. This will turn out to be without loss of generality: for lower bounds, because we will have an access to the flag qubits indicating whether the error in the form of this channel has occurred and thus we will be able to purposely dephase $\regQo$ whenever $\regQi$ dephases; for upper bounds, because the algorithm will work for any oracle that is faultless with probability at least $1/2+\Omega(1)$.

We may omit superscripts $polar$ and $phase$ from $\ccN_1$ when we make general statements addressing both types of noise and in Sections~\ref{sec:envRegs}--\ref{sec:boundProg}, where we only consider the dephasing noise. The subscript $1$ in $\ccN_1$ indicates that the noise causes an error with probability $1$. More generally, for both noise models and for a probability $p\in[0,1]$, we define $\ccN_p$ as a CPTP map that maps $\rho$ to $\rho$ with probability $1-p$ and to $\ccN_1(\rho)$ with probability $p$; in the latter case, we say that the \emph{error occurs}. 
In other words, $\ccN_p:=(1-p)\ccI+p\ccN_1$, where $\ccI$ is the identity map, and it maps
 $\rho$ to $(1-p)\rho + p\ccN_1(\rho)$.
We may refer to $p$ as the \emph{noise level} or the \emph{noise rate}.
We might drop the subscript $p$ from $\ccN_p$ when it is clear from the context, especially when making informal statements. 

We define the noisy oracle call as the CPTP map $\ccO_{f,p}:=\ccN_{p}\circ \ccO_f
= \ccO_f\circ \ccN_{p}$ where $\circ$ denotes the composition of maps.
See Figure~\ref{fig:noisy_orac_comm} for the circuit diagram of the noisy oracle call.

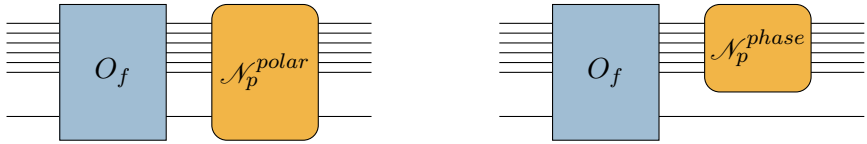
\begin{figure}[!h]
\centering
\begin{tikzpicture}
  \foreach \y in {0,...,5}
      \draw (-2.4,\ygap*\y)--(2.4,\ygap*\y) ;
   \draw (-2.4,-\ygap*4.5)--(2.4,-\ygap*4.5);
   \draw [draw=black,fill=oracleColor] (-1.7,-0.9) rectangle (-0.3,0.9);
   \draw [fill=noiseColor, rounded corners = 2mm] (0.3,-0.9) rectangle (1.7,0.9);
   \node at (-1,0) {$O_f$};
   \node at (1,0) {$\ccN^{polar}_p$};
\end{tikzpicture}
\qquad\qquad
\begin{tikzpicture}
  \foreach \y in {0,...,5}
      \draw (-2.4,\ygap*\y)--(2.4,\ygap*\y) ;
   \draw (-2.4,-\ygap*4.5)--(2.4,-\ygap*4.5);
   \draw [draw=black,fill=oracleColor] (-1.7,-0.9) rectangle (-0.3,0.9);
   \draw [fill=noiseColor, rounded corners = 2mm] (0.3,-0.26) rectangle (1.7,0.9);
   \node at (-1,0) {$O_f$};
   \node at (1,0.32) {$\ccN^{phase}_p$};
\end{tikzpicture}

\captionsetup{font=small}
\captionsetup{width=0.9\textwidth}
\caption[my caption]{%
The circuit diagram of the noisy oracle call corresponding to the depolarizing noise (on the left) and the dephasing noise (on the right). The dephasing noise acts as the identity on the query output register $\regQo$. Here and in other circuit diagrams we use boxes with sharp corners for unitary operations and, more generally, linear isometries, while we use boxes with rounded corners for CPTP maps that are not linear isometries.}
\label{fig:noisy_orac_comm}
\end{figure}

\paragraph{Oracles with flag bits.}

Let us now introduce noisy oracles that signal whether or not the error has occurred. This signal is in a form of a \emph{flag} (qu)bit that the oracle call adds to the workspace of the algorithm. After the oracle call, the algorithm is then permitted to do whatever it pleases with this additional qubit, including, to ignore it, as if it were not even received.  We will refer to this new qubit register, which extends the workspace, as $\regW_+$. Hence, we can formalize \emph{error-signaling probabilistic noise} as a CPTP map that maps from $\regQ$ to $\regQ\regW_+$ and that acts on linear operators $\rho$ on $\regQ$ as
\[
\ccN_p^+\colon \rho\mapsto (1-p)\rho\otimes|\flagQ\>\<\flagQ| + p\ccN_1(\rho)\otimes|\flagC\>\<\flagC|.
\]
Consider a CPTP map $\ccD$ acting on $\regQ\regW_+$ that first measures qubit $\regW_+$ with respect to basis $\{|\flagQ\>,|\flagC\>\}$ and, if the measurement yields $\flagC$, it then applies the completely depolarizing channel $\ccN_1^{polar}$ on $\regQ$.
We have 
$\ccN_p^{polar,+} = \ccD\circ\ccN_p^{phase,+}$.
Since the algorithm can implement $\ccD$ without any queries, any hardness results that we show for error-signaling oracles with dephasing noise will carry over to oracles with depolarizing noise (see also footnote~\ref{fn:targetBasis}).

We define the \emph{error-signaling noisy oracle} call as the CPTP map from $\regQ$ to $\regQ\regW_+$ as the composition $\ccO_{f,p}^+ := \ccN_p^+\circ\ccO_f$. For contrast, we may refer to $\ccO_{f,p}=\ccN_p\circ\ccO_f$, which does not introduce flag bits, as the \emph{error-concealing noisy oracle}.

\subsection{Quantum noisy-query algorithms}
\label{sec:NoisyAlgo}

Now that we have introduced the oracle, let us formalize the computational model of quantum noisy-query algorithm, and describe its operation. We will consider both algorithms with error-concealing oracles and algorithms with error-signaling oracles, which we illustrate in Figure~\ref{fig:algoTwo}.

We divide the memory of the algorithm into two registers: one is the query register $\regQ$ and the other is the \emph{workspace} register $\regW$ that we assume to consist of some number of qubits. As we have described above, in the error-signaling case, every call to the noisy oracle $\ccO^+_{f,p}$ introduces an extra qubit $\regW^+$, which we incorporate into the workspace register.
We call the joint register $\regA=\regQ\regW$ the \emph{algorithm register} (or \emph{registers}).

A quantum noisy-query algorithm is specified by four components: (1) the number of quantum queries, (2) the initial state of the algorithm, (3) input-independent unitary operators that govern the evolution of the quantum system between oracle calls, and (4) the final measurement. Let us describe these components in detail.
\begin{enumerate}
\item
We denote the number of queries by $\tau$ and we enumerate oracle calls from $1$ to $\tau$.
\item Let initially the workspace register $\regW$ consist of $\ell$ qubits, therefore the entire initial memory corresponds to an $nm2^{\ell}$-dimensional Euclidean space. The initial state of the algorithm is an input-independent pure state $|\psi^0\>$ in this space; the first oracle call is performed directly on this state.
\item
In the error-concealing case, the size of the memory does not change throughout the algorithm.
In the error-signaling case,
each noisy oracle call expands the workspace register by a qubit. 
The evolution between oracle calls and after the last call is given by input-independent unitary operators $U_1,\ldots,U_\tau$, where the dimension of $U_t$ is $nm2^{\ell+t}$ in the error-signaling case and $nm2^\ell$ in the error-concealing case.
\item
Given some finite set $\cA$ of answers, the final measurement is given by a set $\{\Pi_a\colon a\in\cA\}$, where each $\Pi_a$ is an orthogonal projector of the same dimension as $U_\tau$, the final unitary operator, and $\sum_a\Pi_a=I$.
\end{enumerate}

\def\xLayerWidth{2.5}
\def\xGateHalfWidth{0.4}
\def\xWirePad{0.5}
\def\xMeasurePad{0.2}
\def\xySlope{\ygap/\xLayerWidth}

\begin{figure}[h!]
\centering
\begin{tikzpicture}
\foreach \y in {0,...,5}
   \draw (0.5*\xLayerWidth-\xGateHalfWidth-\xWirePad,\ygap*\y) -- 
             (4.0*\xLayerWidth+\xGateHalfWidth+\xWirePad+\xMeasurePad,\ygap*\y);
\draw (0.5*\xLayerWidth-\xGateHalfWidth-\xWirePad,-\ygap*4.5) --
          (4.0*\xLayerWidth+\xGateHalfWidth+\xWirePad,-\ygap*4.5);

\draw [fill=gray!40] (4.0*\xLayerWidth+\xGateHalfWidth+\xWirePad+\xMeasurePad,-0.3) rectangle
          (4.0*\xLayerWidth+\xGateHalfWidth+\xWirePad+\xMeasurePad+2*\xGateHalfWidth,1);
\draw[->] (4.0*\xLayerWidth+\xGateHalfWidth+\xWirePad+\xMeasurePad+\xGateHalfWidth,0.2) -- (4.0*\xLayerWidth+\xGateHalfWidth+\xWirePad+\xMeasurePad+1.5*\xGateHalfWidth,0.7);      
\draw (4.0*\xLayerWidth+\xGateHalfWidth+\xWirePad+\xMeasurePad+1.8*\xGateHalfWidth,0.4) arc
    [
        start angle=40,
        end angle=140,
        x radius=0.4 cm,
        y radius =0.4 cm
    ] ;

\foreach \y in {0,...,7}
   \draw (0.5*\xLayerWidth-\xGateHalfWidth-\xWirePad,
                  -1-\ygap*\y+\xGateHalfWidth*\xySlope
                            + \xWirePad*\xySlope-\xLayerWidth*\xySlope ) --
             (4.0*\xLayerWidth+\xGateHalfWidth+\xWirePad,
                  -1-\ygap*\y+\xGateHalfWidth*\xySlope
                            + \xWirePad*\xySlope-\xLayerWidth*\xySlope ) ;
    
\foreach \i in {1,...,4}
{
   \draw [draw=black,fill=algoColor, opacity=1]  (\xLayerWidth*\i-\xGateHalfWidth,-2.2) rectangle (\xLayerWidth*\i+\xGateHalfWidth,1);
   \node at (\xLayerWidth*\i,-0.5) {$U_\i$};
  \draw [top color=noiseColor!55!oracleColor, bottom color=oracleColor, rounded corners = 2mm, opacity = 1]   (\xLayerWidth*\i-0.5*\xLayerWidth-\xGateHalfWidth,-1) rectangle 
  (\xLayerWidth*\i-0.5*\xLayerWidth+\xGateHalfWidth,1);
   \node at (\xLayerWidth*\i-0.5*\xLayerWidth,0) {$\ccO_{f,p}$};
}

\node at (-0.5,-0.65) {$|\psi^0\>$};
\draw [decorate,
    decoration = {calligraphic brace, amplitude=5pt}] (0.2,-2.15) --  (0.2,0.8);
    
\end{tikzpicture}
{\small(\emph{error-concealing case})}

%
%

\bigskip

\begin{tikzpicture}
\foreach \y in {0,...,5}
   \draw (0.5*\xLayerWidth-\xGateHalfWidth-\xWirePad,\ygap*\y) -- 
             (4.0*\xLayerWidth+\xGateHalfWidth+\xWirePad+\xMeasurePad,\ygap*\y);
\draw (0.5*\xLayerWidth-\xGateHalfWidth-\xWirePad,-\ygap*4.5) --
          (4.0*\xLayerWidth+\xGateHalfWidth+\xWirePad,-\ygap*4.5);

\draw [fill=gray!40] (4.0*\xLayerWidth+\xGateHalfWidth+\xWirePad+\xMeasurePad,-0.3) rectangle
          (4.0*\xLayerWidth+\xGateHalfWidth+\xWirePad+\xMeasurePad+2*\xGateHalfWidth,1);
\draw[->] (4.0*\xLayerWidth+\xGateHalfWidth+\xWirePad+\xMeasurePad+\xGateHalfWidth,0.2) -- (4.0*\xLayerWidth+\xGateHalfWidth+\xWirePad+\xMeasurePad+1.5*\xGateHalfWidth,0.7);      
\draw (4.0*\xLayerWidth+\xGateHalfWidth+\xWirePad+\xMeasurePad+1.8*\xGateHalfWidth,0.4) arc
    [
        start angle=40,
        end angle=140,
        x radius=0.4 cm,
        y radius =0.4 cm
    ] ;

\foreach \i in {1,...,4}
{
   \draw (\i*\xLayerWidth-\xGateHalfWidth,
                       -1-0.5*\xLayerWidth*\xySlope+\xGateHalfWidth*\xySlope) -- 
             (4.0*\xLayerWidth+\xGateHalfWidth+\xWirePad,
             -1+\ygap*\i -4.5*\xLayerWidth*\xySlope-\xGateHalfWidth*\xySlope
                            -\xWirePad*\xySlope);                            
   \draw (\i*\xLayerWidth-0.5*\xLayerWidth+\xGateHalfWidth,
                             -1-\xGateHalfWidth*\xySlope+\ygap) .. controls 
             (\i*\xLayerWidth-0.25*\xLayerWidth,
                       -1-0.25*\xLayerWidth*\xySlope+\ygap) and
             (\i*\xLayerWidth-0.25*\xLayerWidth,
                       -1-0.25*\xLayerWidth*\xySlope) ..
             (\i*\xLayerWidth-\xGateHalfWidth,
                       -1-0.5*\xLayerWidth*\xySlope+\xGateHalfWidth*\xySlope);
}                           

\foreach \y in {0,...,7}
   \draw (0.5*\xLayerWidth-\xGateHalfWidth-\xWirePad,
                  -1-\ygap*\y+\xGateHalfWidth*\xySlope
                            + \xWirePad*\xySlope-\xLayerWidth*\xySlope ) --
             (4.0*\xLayerWidth+\xGateHalfWidth+\xWirePad,
                  -1-\ygap*\y-4.5*\xLayerWidth*\xySlope-\xGateHalfWidth*\xySlope
                            -\xWirePad*\xySlope) ;
    
\foreach \i in {1,...,4}
{
   \draw [draw=black,fill=algoColor, opacity=1]  (\xLayerWidth*\i-\xGateHalfWidth,-2.2-\i*\ygap) rectangle (\xLayerWidth*\i+\xGateHalfWidth,1);
   \node at (\xLayerWidth*\i,-0.5-0.5*\i*\ygap) {$U_\i$};
  \draw [top color=noiseColor!55!oracleColor, bottom color=oracleColor, rounded corners = 2mm, opacity = 1]   (\xLayerWidth*\i-0.5*\xLayerWidth-\xGateHalfWidth,-0.9) -- 
  (\xLayerWidth*\i-0.5*\xLayerWidth-\xGateHalfWidth,1) --
  (\xLayerWidth*\i-0.5*\xLayerWidth+\xGateHalfWidth,1) --
  (\xLayerWidth*\i-0.5*\xLayerWidth+\xGateHalfWidth,-1.1) -- cycle;
   \node at (\xLayerWidth*\i-0.5*\xLayerWidth,0) {$\ccO_{f,p}^+$};
}

\node at (-0.5,-0.65) {$|\psi^0\>$};
\draw [decorate,
    decoration = {calligraphic brace, amplitude=5pt}] (0.2,-2.15) --  (0.2,0.8);
    
\end{tikzpicture}
{\small(\emph{error-signaling case})}
\captionsetup{font=small}
\captionsetup{width=0.9\textwidth}
\caption[my caption]{%
A quantum query algorithm with error-concealing noisy oracle (top) and with error-signaling noisy oracle (bottom). Both algorithms make $\tau=4$ queries, and their initial workspace consists of $\ell=8$ qubits. For the final measurement, here we focus on the problem of unstructured search, for which the set of answers is $\cA=[n]$ and, without loss of generality, the final measurement is given by the set $\{|x\>\<x|_\regQi\otimes I_{\regQo\regW}\colon x\in[n]\}$. Therefore we illustrate the final measurement to act only on $\regQi$.}
\label{fig:algoTwo}
\end{figure}
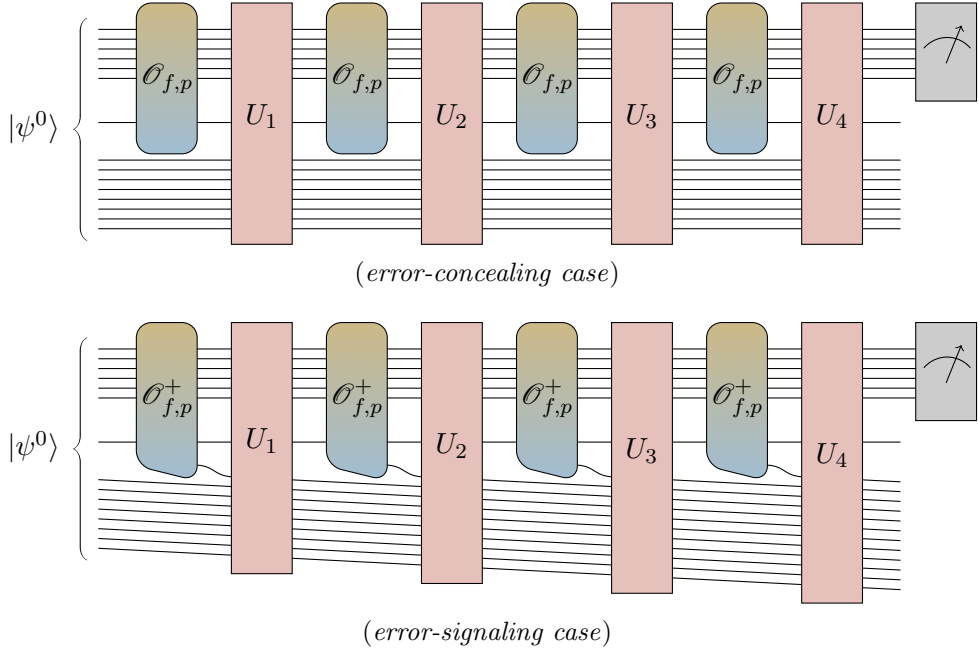

The execution of the algorithm starts in the initial state $|\psi^0\>$, and then alternates between oracle calls and input-independent unitaries as follows. Iteratively, for $t=1,2,\ldots,\tau$, the algorithm first performs an oracle call ($\ccO_{f,p}$ or $\ccO_{f,p}^+$), and then applies unitary $U_{t}$. Finally, the algorithm performs the measurement according to $\{\Pi_a\colon a\in\cA\}$, returning the measurement outcome $a$ as an answer. We say that the algorithm is \emph{successful} if $a$ is a correct answer for $f$, and we say that it \emph{fails} otherwise.
See Figure~\ref{fig:algoTwo} for circuit diagrams of quantum query algorithms with error-concealing and error-signaling noisy oracles.

Note that, because the workspace is not affected by the noise, the model that we are considering is equivalent to the one which allows intermediate measurements and classical post-precessing.

\subsection{Unstructured Search and Hardest Instances}
\label{sec:SearchDef}

From now on, let us focus on the problem of unstructured search. We will consider two of its versions: the worst case, in the main text, and the average case, in the appendix.

\paragraph{Worst case.}

For the worst case scenario, it suffices to consider the case when $m=2$. Here, for a function $f\colon[n]\rightarrow\{0,1\}$, we call an input $x\in[n]$ \emph{marked} if $f(x)=1$. The goal of the unstructured search is, given an oracle access to $f$, to find a marked input $x$, assuming there exists one.

Intuitively, the hardest instances of the problem are functions $f$ that have exactly one marked input, and, for lower bounds, we will only consider such functions. Namely, let $f_x\colon[n]\rightarrow\{0,1\}$ be the function for which $x$ is the unique marked input, that is, 
\[
f_x(x')=
\begin{cases}
1 & \text{if }x'=x, \\
0 & \text{if }x'\in[n]\setminus\{x\}.
\end{cases}
\]
Note that the noiseless oracle corresponding to $f_x$ is $O_{f_x} = I_{2n} - 2|x,-\>\<x,-|$.

We define the $\epsilon$-error quantum noisy-query complexity of the unstructured search to be the minimum number of queries needed by any algorithm that errs with probability no more than $\epsilon$ for all $f$ having at least one marked element (i.e., all $f$ other than the constant $f(x)=0$).

\paragraph{Average case.}
In the average case, for a function $f\colon[n]\rightarrow[m]$, we call an input $x\in[n]$ \emph{marked} if, in binary, $f(x)=\marked$, where $\marked$ is the all-zeros string $0^{\log m}$. Here, the algorithm is given an access to such an $f$ chosen uniformly at random, each with probability $1/m^n$.

It is possible that $f$ contains no marked elements, in which case the algorithm is guaranteed to fail. Therefore, one typically considers the scenario where $m=\OO(n)$ so that at least a constant fraction of functions $f$ have at least one marked element.

\section{Purifying the Computation}
\label{sec:envRegs}

At this point, we start working towards the proof of our main result, Theorem~\ref{thm:main}, the whole proof covering Sections~\ref{sec:envRegs}--\ref{sec:boundProg}.
In this section we introduce registers that purify the overall computation, those extra registers being a part of the overall lower bound framework.
From now on, until  Section~\ref{sec:algoSketch}, we will only consider the error-signaling noisy oracle corresponding to the dephasing noise, and, for brevity, we will simply refer to it as the noisy oracle.

Let us consider a scenario where we are given a noisy oracle access to the function $f_x$ being chosen uniformly at random from the set of functions $\Func:=\{f_1,\ldots,f_n\}$, that is, from the set of functions with exactly one marked element. Our goal is to find $x$.

Without loss of generality, let the query-input register $\regQi$ be also used by the algorithm to return the answer (as already illustrated in Figure~\ref{fig:algoTwo}), and thus let $\{ \Pi_x:=|x\>\<x| \colon x\in[n]\}$ be the final measurement. Hence, if $\rho_{x,\tau}$ is the final state of the computation with oracle access to $f_x$, then the average success probability $q_{succ}$ is given by 
\[
q_{succ} = \frac{1}{n}\sum_{x\in[n]}\tr[(\Pi_x\otimes I_{\regQo\regW})\rho_{x,\tau}].
\]

\subsection{Environment registers}

The overall quantum system will consist of three sets of registers. In addition to the registers $\regQ$ and $\regW$ used by the algorithm, we introduce two additional registers, which we describe below.

\paragraph{Truth register.}

The \emph{truth-table} register (or, the \emph{truth} register, for short) $\regFunc$  corresponds to the set of functions $\Func$. One can think of each its basis state $|f\>$ containing the full truth table of $f$.%
\footnote{In similar lower bound techniques, this register is called the oracle register, the adversary register, and the function register.}

At the beginning of the computation, we initialize the register $\regFunc$ as the uniform superposition $|\unif\>:=\sum_{x\in[n]}|f_x\>/\sqrt{n}$, and we can think of this register as effectively purifying the random choice of $f_x\in\Func$.

\paragraph{Record register.} 

The \emph{record} register $\regR$ will start empty, then the subregister $\regR_1$ will be appended to it by the first oracle call, $\regR_2$ by the second oracle call, and so on. 
For every $t\in\{1,\ldots,\tau\}$, the register $\regR_t$ corresponds to the set $\rec:=\{\bot,1,2,\ldots,n\}=[n]\cup\{\bot\}$,
where $\bot$ indicates the absence of error in $t$-th oracle call,%
\footnote{The state $|\bot\>$ corresponding to the label $\bot\in\rec$ is unrelated to an equally-denoted state in Zhandry's work on the compressed oracle~\cite{zhandry:record}.}
while, $x\in[n]$ indicates that the dephasing has occurred and the query-input register has been dephased to the basis state $|x\>$.
In effect, the register $\regR_t$ will purify the action of non-unitary, noisy oracle call number $t$.

After $t$ oracle calls, the register $\regR$ corresponds to the set $\rec^t$, and we think about its contents as strings $R=r_1r_2\ldots r_t$ of length $t$, which we call \emph{records}. As per standard notation, let $\rec^*:=\bigcup_t \rec^t$ be the set of all records.

For a record $R=r_1\ldots r_t\in\rec^*$, we say that an input $x\in[n]$ \emph{appears} in $R$ if $r_{t'}=x$ for some $t'$, and let $\inrec{R}\subset[n]$ be the set of all inputs appearing in $R$; note that here we ignore the symbol $\bot$. Also note that, for $R\in \rec^t$, we have $|\inrec{R}|\le t$.
For conciseness, let $n_R:=n-|\inrec{R}| = |[n]\setminus\inrec{R}|$, which will be approximately $n$ because we will mostly consider $R\in\rec^t$ with $t\ll n$.

For a record $R=r_1\ldots r_t\in\rec^t$ and a symbol $r\in \rec$, let $R\app r$ denote the record obtained by appending $r$ to $R$. That is, $R\app r=r_1\ldots r_tr\in\rec^{t+1}$. We have $\inrec{(R\app\bot)}=\inrec{R}$ and $\inrec{(R\app x)}=\inrec{R}\cup\{x\}$ for $x\in[n]$.

\subsection{Extended oracles}

We define the extended noisy oracle call as a linear isometry 
\[
O_p :=\sum_{f_z\in\Func}|f_z\>\<f_z|_\regFunc\otimes\Big(\big(\sqrt{1-p}I_\regQi\otimes|\flagQ,\bot\>
+\sqrt{p}\sum_{x\in[n]}|x\>\<x|\otimes|\flagC,x\>\big)\otimes I_\regQo\Big)O_{f_z},
\]
where the states $|\flagQ,\bot\>$ and $|\flagC,x\>$ are on registers $\regW_+\regR_t$. Here, the subscript of the register $\regR_t$ indicates that we implicitly think of $O_p$ as $t$-th oracle call.
See Figure~\ref{fig:purifiedOrac} for one potential circuit diagram for $O_p$.
Note that, if we call $O_p$ with the truth register $\regFunc$ being in the state $|f_z\>$ and then discard the newly introduced record register $\regR_t$, then the resulting map is equal to $\ccO_{f_z,p}^+$.

\def\ygap{0.16}
\def\xgap{0.20} 
\def\rCXctrl{0.07}
\def\rCXnot{0.11}

\def\yRecordTop{4.7}
\def\yFlagRec{3.7}
\def\yTruthTop{3.2}
\def\yIndexTop{2}
\def\yTarget{1}
\def\yFlagAlgo{0.41}

\def\xWireLeft{0}
\def\xWireMid{3.4}
\def\xWireRight{6.42}

\def\xOracle{1}
\def\yOraclePad{0.2} 
\def\xFlagCopy{4.2}
\def\xIndexCopy{5.0}
\def\xKetEnd{3.35} 

\begin{figure}[!h]
\centering
\begin{tikzpicture}

\draw (\xWireMid,\yFlagRec)--(\xWireRight,\yFlagRec);
\draw (\xWireLeft,\yTarget)--(\xWireRight,\yTarget);
\draw (\xWireMid,\yFlagAlgo)--(\xWireRight,\yFlagAlgo);
\foreach \i in {0,...,3} {
     \draw (\xWireMid,\yRecordTop-\ygap*\i)--(\xWireRight,\yRecordTop-\ygap*\i);
     \draw (\xWireLeft,\yTruthTop-\ygap*\i)--(\xWireRight,\yTruthTop-\ygap*\i);
     \draw (\xWireLeft,\yIndexTop-\ygap*\i)--(\xWireRight,\yIndexTop-\ygap*\i);
     \draw (\xIndexCopy+\xgap*\i,\yRecordTop-\ygap*\i+\rCXnot) -- (\xIndexCopy+\xgap*\i,\yIndexTop-\ygap*\i);
     \draw (\xIndexCopy+\xgap*\i,\yRecordTop-\ygap*\i) circle (\rCXnot);
     \draw [fill = black] (\xIndexCopy+\xgap*\i,\yFlagRec) circle (\rCXctrl);
     \draw [fill = black] (\xIndexCopy+\xgap*\i,\yIndexTop-\ygap*\i) circle (\rCXctrl);
}
\draw (\xFlagCopy,\yFlagRec) -- (\xFlagCopy,\yFlagAlgo-\rCXnot);
\draw [fill = black] (\xFlagCopy,\yFlagRec) circle (\rCXctrl);
\draw (\xFlagCopy,\yFlagAlgo) circle (\rCXnot);
\draw (\xOracle,\yTruthTop) -- (\xOracle,\yIndexTop);
\draw [fill = black] (\xOracle,\yTruthTop) circle (\rCXctrl);
\draw [fill = black] (\xOracle,\yTruthTop-3*\ygap) circle (\rCXctrl);
\draw [fill = black] (\xOracle-\rCXctrl,\yTruthTop-3*\ygap) rectangle (\xOracle+\rCXctrl,\yTruthTop);
\node at (\xOracle+0.12,\yTruthTop-3*\ygap-0.17) {${}_f$};
\draw [draw=black,fill=oracleColor] (\xOracle-\xGateHalfWidth,\yTarget-\yOraclePad) rectangle (\xOracle+\xGateHalfWidth,\yIndexTop+\yOraclePad);
\node at (\xOracle,0.5*\yTarget+0.5*\yIndexTop) {$O_f$};
\draw [dotted, thick, orange] (\xWireLeft-3.7,0.5*\yFlagRec+0.5*\yTruthTop)--(\xWireRight+0.25,0.5*\yFlagRec+0.5*\yTruthTop);
\draw [dotted, thick, orange] (\xWireLeft-3.7,0.5*\yIndexTop+0.5*\yTruthTop-1.5*\ygap)--(\xWireRight+0.25,0.5*\yIndexTop+0.5*\yTruthTop-1.5*\ygap);
\node [left] at (\xKetEnd,\yRecordTop-1.5*\ygap) {$|0\>^{\otimes\log n}$};
\node [left] at (\xKetEnd,\yFlagRec) {$\sqrt{1-p}|0\>+\sqrt{p}|1\>$};
\node [left] at (\xKetEnd,\yFlagAlgo) {$|0\>$};
\node [] at (-3,\yTruthTop-1.5*\ygap+1.1) {\small record};
\node [] at (-3,\yTruthTop-1.5*\ygap) {\small truth-table};
\node [] at (-3,\yTruthTop-1.5*\ygap-1.1) {\small algorithm};
\node [right] at (-1,0.5*\yRecordTop+0.5*\yFlagRec) {$\regR_t$};
\node [right] at (-1,\yTruthTop-1.5*\ygap) {$\regFunc$};
\node [right] at (-1,\yIndexTop-1.5*\ygap) {$\regQi$};
\node [right] at (-1,\yTarget) {$\regQo$};
\node [right] at (-1,\yFlagAlgo) {$\regW^+$};

\end{tikzpicture}
\captionsetup{font=small}
\captionsetup{width=0.9\textwidth}
\caption[my caption]{%
A circuit diagram of $O_p$, which can be thought of as a controlled version of $O_f$, namely, $\sum_{z\in[n]}|f_z\>\<f_z|_\regFunc\otimes (O_{f_z})_\regQ$, followed by a purification (essentially, as a Stinespring representation) of the error-signaling dephasing noise $\ccN_p^+$.  The symbols of $\rec$ are encoded in the record register $\regR_t$ as follows: $|0_{\mathrm b},0\>$ represents $\bot$ while $|x_{\mathrm b},1\>$ represents $x\in[n]$, where $x_{\mathrm b}$ is the $\log n$-bit binary encoding of $x$. Note that the completely dephasing channel on $\regQi$ is equivalent to copying the state of that register, in the computational basis, to a fresh ancilla, and then disregarding that ancilla.
As mentioned in Section~\ref{sec:QuantMemo}, for the flag $\regW^+$, we associate $0=\flagQ$ and $1=\flagC$.}
\label{fig:purifiedOrac}
\end{figure}
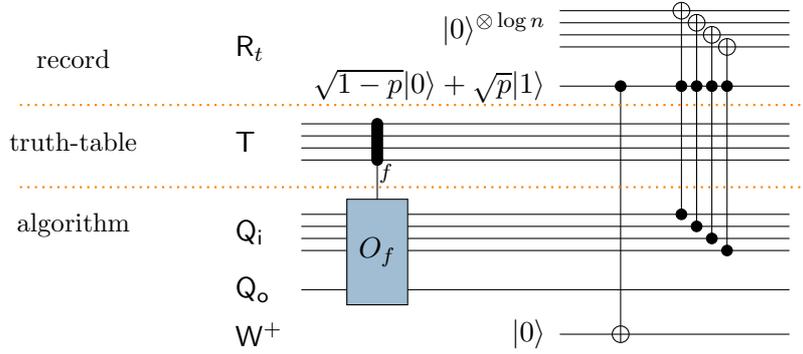

While the $t$-th extended noisy oracle call $O_p$ acts as the identity on earlier record subregisters $\regR_1\ldots \regR_{t-1}$, it is sometimes useful to reintroduce them---and thus the whole record register $\regR$---in the notation when expressing $O_p$. So, we can $t$-th call $O_p$ as 
\begin{multline*}
O_p =\sum_{f_z\in\Func}|f_z\>\<f_z|\otimes\sum_{R\in\rec^{t-1}}\Big(\big(\sqrt{1-p}I_\regQi\otimes|\flagQ,R\app\bot\>\<R|
\\ + \sqrt{p}\sum_{x\in[n]}|x\>\<x|\otimes|\flagC,R\app x\>\<R|\big)\otimes I_\regQo\Big)O_{f_z}.
\end{multline*}

It will be useful to separate the errorless and the erroneous components of $O_p$, and express it as $O_p=\sqrt{1-p}O_Q+\sqrt{p}O_C$,
where
\begin{align*}
& O_Q := O_0 = \sum_{f}|f\>\<f|\otimes O_f \otimes|\flagQ,\bot\>,
\\ &
O_C := O_1 =\sum_{f}|f\>\<f|\otimes \sum_{x\in[n]} |x\>\<x|
\otimes \sum_{y\in\{0,1\}}|y\oplus f(x)\>\<y| \otimes|\flagC,x\>
\end{align*}
are linear isometries with orthogonal images.
The latter one, $O_C$, exhibits close similarity to the classical oracle considered in works \cite{Rosmanis:2022:hybrid,hamoudi:2022:tradeoffs} on hybrid quantum-classical query algorithms, except that it permits the query-output register to remain in a superposition. For that reason, we may refer to $O_C$ as the \emph{classical oracle} and to $O_Q$ as the \emph{quantum oracle}.

\subsection{Extended computation}

Now let us consider how the algorithm extends to the entire space of registers $\regFunc\regA\regR$, and describe its execution.
The execution of the \emph{(extended) computation} starts in the initial state $|\phi_0\>:=|\unif\>\otimes|\psi^0\>$, which is also the state of the overall system just before the first oracle call. 
Recall that $|\unif\>=\sum_{z\in[n]}|f_z\>/\sqrt{n}$ and also recall that we start the computation with empty record, that is, initially, the register $\regR$ corresponds to one-dimensional Euclidean space (i.e., zero qubits).
Next, similarly as before, the computation alternates between extended noisy oracle calls on $\regFunc\regQ\regR$ and input-independent unitaries on $\regQ\regW$ as follows. Iteratively, for $t=1,2,\ldots,\tau$, the computation first performs an oracle call $O_{p}$, and then applies unitary $U_{t}$.
Finally, on the final state, the extended computation measures registers $\regFunc$ and $\regQi$, obtaining a function $f_z$ and an input $x$, and the computation is successful if and only if $x$ is a marked input for $f_z$, that is, $x=z$. Accordingly, let us define the projector on successful outcomes as
 $\Pi_{succ}:=\sum_{x\in[n]}|f_x,x\>\<f_x,x|$, which acts on registers $\regFunc\regQi$.
 
Because both $U_t$ and $O_p$ are linear isometries, the overall memory of the extended computation always remain in a pure state. For $t\in\{1,\ldots,\tau-1\}$, let $|\phi_t\>$ be the state of the overall system just before $(t+1)$-th oracle call and let $|\phi_\tau\>$ be the final state of the system. These states can be recursively expressed as $|\phi_t\>=U_t O_p|\phi_{t-1}\>$ for all $t\in\{1,\ldots,\tau\}$. Because discarding registers $\regFunc$ and $\regR$ from the computation would be equivalent to running the algorithm on a randomly chosen $f_z$, we have $q_{succ}=\|\Pi_{succ}|\phi_\tau\>\|^2$.

See Figure~\ref{fig:purifiedAlgo} for the circuit diagram of the extended computation. (Note that,  for the sake of simplicity, in the illustration, the order of registers is $\regR\regFunc\regA$, while, in formulae, we typically use the order $\regFunc\regA\regR$.)

\def\xLayerWidth{2.5}
\def\xGateHalfWidth{0.4}
\def\xWirePad{0.5}
\def\xMeasurePad{0.2}
\def\xySlope{\ygap/\xLayerWidth}
\def\xyHalfSlope{0.5*\xySlope}

\begin{figure}[!h]
\centering
\begin{tikzpicture}
\draw [ultra thick] (0.5*\xLayerWidth-\xGateHalfWidth-\xWirePad,1.2) --
          (4.0*\xLayerWidth+\xGateHalfWidth+\xWirePad+\xMeasurePad,1.2);
   \draw [ultra thick] (0.5*\xLayerWidth-\xGateHalfWidth-\xWirePad,0.2) -- 
             (4.0*\xLayerWidth+\xGateHalfWidth+\xWirePad+\xMeasurePad,0.2);
\draw (0.5*\xLayerWidth-\xGateHalfWidth-\xWirePad,-0.35) --
          (4.0*\xLayerWidth+\xGateHalfWidth+\xWirePad,-0.35);

\draw [dotted, thick, orange] (-2.7+0.5*\xLayerWidth-\xGateHalfWidth-\xWirePad,1.7) --
          (0.8+4.0*\xLayerWidth+\xGateHalfWidth+\xWirePad,1.7);
\draw [dotted, thick, orange] (-2.7+0.5*\xLayerWidth-\xGateHalfWidth-\xWirePad,0.7) --
          (0.8+4.0*\xLayerWidth+\xGateHalfWidth+\xWirePad,0.7);
\node [] at (-1.7,2.2) {\small record};
\node [] at (-1.7,1.2) {\small truth-table};
\node [] at (-1.7,0.2) {\small algorithm};

\draw [fill=gray!40] (4.0*\xLayerWidth+\xGateHalfWidth+\xWirePad+\xMeasurePad,
                      1.2-\xGateHalfWidth) rectangle
          (4.0*\xLayerWidth+\xGateHalfWidth+\xWirePad+\xMeasurePad+2*\xGateHalfWidth,
                      1.2+\xGateHalfWidth);
\node [right] at (4.0*\xLayerWidth+\xGateHalfWidth+\xWirePad+\xMeasurePad+2*\xGateHalfWidth-0.45,
                      1.39-\xGateHalfWidth) {\small$f\!{}_z$};

\draw [fill=gray!40] (4.0*\xLayerWidth+\xGateHalfWidth+\xWirePad+\xMeasurePad,0.2-\xGateHalfWidth) rectangle
          (4.0*\xLayerWidth+\xGateHalfWidth+\xWirePad+\xMeasurePad+2*\xGateHalfWidth,0.2+\xGateHalfWidth);
\node [right] at (4.0*\xLayerWidth+\xGateHalfWidth+\xWirePad+\xMeasurePad+2*\xGateHalfWidth-0.39,
                      0.34-\xGateHalfWidth) {\small$x$};                      

\draw[->] (4.0*\xLayerWidth+\xGateHalfWidth+\xWirePad+\xMeasurePad+\xGateHalfWidth,0.05) -- (4.0*\xLayerWidth+\xGateHalfWidth+\xWirePad+\xMeasurePad+1.5*\xGateHalfWidth,0.55);      
\draw (4.0*\xLayerWidth+\xGateHalfWidth+\xWirePad+\xMeasurePad+1.8*\xGateHalfWidth,0.25) arc
    [
        start angle=40,
        end angle=140,
        x radius=0.4 cm,
        y radius =0.4 cm
    ] ;                          

\draw[->] (4.0*\xLayerWidth+\xGateHalfWidth+\xWirePad+\xMeasurePad+\xGateHalfWidth,1.05) -- (4.0*\xLayerWidth+\xGateHalfWidth+\xWirePad+\xMeasurePad+1.5*\xGateHalfWidth,1.55);      
\draw (4.0*\xLayerWidth+\xGateHalfWidth+\xWirePad+\xMeasurePad+1.8*\xGateHalfWidth,1.25) arc
    [
        start angle=40,
        end angle=140,
        x radius=0.4 cm,
        y radius =0.4 cm
    ] ;

\foreach \i in {1,...,4}
{
   \draw [ultra thick] (\i*\xLayerWidth-\xGateHalfWidth,
                       2+0.5*\xLayerWidth*\xySlope-\xGateHalfWidth*\xySlope) -- 
             (4.0*\xLayerWidth+\xGateHalfWidth+\xWirePad,
             2-\ygap*\i +4.5*\xLayerWidth*\xySlope+\xGateHalfWidth*\xySlope
                            +\xWirePad*\xySlope);        
   \draw [ultra thick] (\i*\xLayerWidth-0.5*\xLayerWidth+\xGateHalfWidth,
                             2+\xGateHalfWidth*\xySlope-\ygap) .. controls 
             (\i*\xLayerWidth-0.25*\xLayerWidth,
                       2+0.25*\xLayerWidth*\xySlope-\ygap) and
             (\i*\xLayerWidth-0.25*\xLayerWidth,
                       2+0.25*\xLayerWidth*\xySlope) ..
             (\i*\xLayerWidth-\xGateHalfWidth,
                       2+0.5*\xLayerWidth*\xySlope-\xGateHalfWidth*\xySlope);
   \draw (\i*\xLayerWidth-\xGateHalfWidth,
                       -0.7-0.5*\xLayerWidth*\xyHalfSlope+\xGateHalfWidth*\xyHalfSlope) -- 
             (4.0*\xLayerWidth+\xGateHalfWidth+\xWirePad,
             -0.7+0.5*\ygap*\i -4.5*\xLayerWidth*\xyHalfSlope-\xGateHalfWidth*\xyHalfSlope
                            -\xWirePad*\xyHalfSlope);
   \draw (\i*\xLayerWidth-0.5*\xLayerWidth+\xGateHalfWidth,
                             -0.7-\xGateHalfWidth*\xyHalfSlope+\ygap) .. controls 
             (\i*\xLayerWidth-0.25*\xLayerWidth,
                       -0.7-0.25*\xLayerWidth*\xyHalfSlope+\ygap) and
             (\i*\xLayerWidth-0.25*\xLayerWidth,
                       -0.7-0.25*\xLayerWidth*\xyHalfSlope) ..
             (\i*\xLayerWidth-\xGateHalfWidth,
                       -0.7-0.5*\xLayerWidth*\xyHalfSlope+\xGateHalfWidth*\xyHalfSlope);
}                           

\foreach \y in {0,...,7}
   \draw (0.5*\xLayerWidth-\xGateHalfWidth-\xWirePad,
                  -0.7-0.5*\ygap*\y+\xGateHalfWidth*\xyHalfSlope
                            + \xWirePad*\xyHalfSlope-\xLayerWidth*\xyHalfSlope ) --
             (4.0*\xLayerWidth+\xGateHalfWidth+\xWirePad,
                  -0.7-0.5*\ygap*\y-4.5*\xLayerWidth*\xyHalfSlope-\xGateHalfWidth*\xyHalfSlope
                            -\xWirePad*\xyHalfSlope) ;
    
\foreach \i in {1,...,4}
{
   \draw [draw=black,fill=algoColor]  (\xLayerWidth*\i-\xGateHalfWidth,-1.65-0.5*\i*\ygap) rectangle (\xLayerWidth*\i+\xGateHalfWidth,0.2+\xGateHalfWidth);
   \node at (\xLayerWidth*\i,-0.5-0.25*\i*\ygap) {$U_\i$};
  \draw [top color=noiseColor!55!oracleColor, bottom color=oracleColor]   (\xLayerWidth*\i-0.5*\xLayerWidth-\xGateHalfWidth,-0.6) -- 
  (\xLayerWidth*\i-0.5*\xLayerWidth-\xGateHalfWidth,1.97) --
  (\xLayerWidth*\i-0.5*\xLayerWidth+\xGateHalfWidth,2.07) --
  (\xLayerWidth*\i-0.5*\xLayerWidth+\xGateHalfWidth,-0.7) -- cycle;
   \node at (\xLayerWidth*\i-0.5*\xLayerWidth,0.55) {$O_p$};
}

\node [left] at (0.1, 1.2) {$|\unif\>$};
\node [left] at (0.1,-0.65) {$|\psi^0\>$};
\draw [decorate,
    decoration = {calligraphic brace, amplitude=5pt}] (0.25,-1.65) --  (0.25,0.3);

\foreach \i in {0,...,4}
{
   \draw [thick, phiStateCol, opacity=0.35] (0.25*\xLayerWidth+\xLayerWidth*\i,-1.63-0.5*\i*\ygap) -- (0.25*\xLayerWidth+\xLayerWidth*\i,1.73+\xGateHalfWidth+\i*\ygap);
   \node [below, phiStateCol, opacity=0.9] at (0.1+0.25*\xLayerWidth+\xLayerWidth*\i,-1.63-0.5*\i*\ygap) {$|\phi_\i\>$};
}
    
\end{tikzpicture}
\captionsetup{font=small}
\captionsetup{width=0.9\textwidth}
\caption[my caption]{%
The extended computation with four extended oracle calls. The state of the overall memory remains pure throughout the computation, with $|\phi_t\>$ being the state just after the unitary $U_t$. The thick wires represent multiple qubits bundled in a single register, corresponding to sets $\rec$, $\Func$, or $[n]$. At the end of the computation, registers $\regFunc$ and $\regQi$ are measured, yielding $f_z$ and $x$, respectively, and the computation is successful if $z=x$.}
\label{fig:purifiedAlgo}
\end{figure}
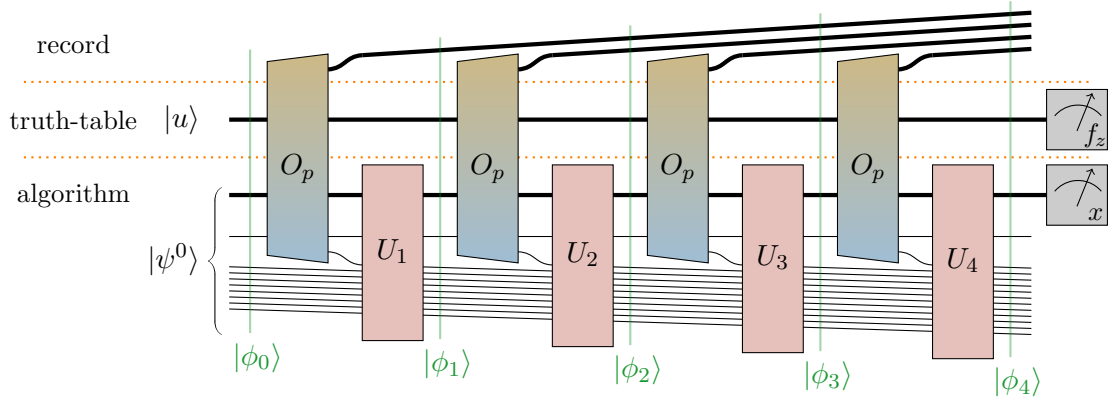

\section{Progress Measure}
\label{sec:progressMes}

In this section, we introduce a measure to quantify the progress of the computation towards solving the search problem. The value of this measure will be based on $|\phi_t\>$, the current state of the computation. Here we also present our main lemma, Lemma~\ref{lem:progEvol}, which addresses the initial and the final values of the progress measure, as well as bounds by how much a single oracle call can increase the progress measure. Then, theorem~\ref{thm:main} is a relatively straightforward corollary of Lemma~\ref{lem:progEvol}.

To start with, let us decompose the whole joint space of registers $\regFunc\regA\regR$ into three subspaces corresponding to the scenarios where, informally, 1) the query-input register has been dephased to a marked input by the classical oracle $O_C$, 2) the query-input register has not been dephased to the marked input, but the quantum oracle $O_Q$ has found the marked input, and 3) no progress has been made. When considering oracle calls in Section~\ref{sec:transit}, we will further decompose the second subspace.

\subsection{Progress-defining subspaces}

Let us consider the $n$-dimensional space corresponding to the register $\regFunc$, that is, the space 
spanned by $|f_x\>$. Recall that the initial state of the truth register is the uniform superposition $|\unif\>=\sum_{x\in[n]}|f_x\>/\sqrt{n}$.
Also recall the notation $n_R=n-|\inrec{R}|$.
Let us define the unit vector
\[
|\psi_{\inrec{R}}\>:= \sum_{x\in[n]\setminus \inrec{R}} |f_x\>/\sqrt{n_R},
\]
for which we have
\begin{equation}
\label{eq:fxPsiR}
\< f_x|\psi_{\inrec{R}}\> = 
\begin{cases}
1/\sqrt{n_R} & \text{if } x \in [n]\setminus\inrec{R}, \\
0 & \text{if } x \in \inrec{R}.
\end{cases}
\end{equation}
In turn, let us define projectors
\begin{align*}
&
\Pi^{\progC}_{t} := \sum_{R\in\rec^t} \sum_{z\in\inrec{R}}|f_z\>\<f_z| \otimes |R\>\<R|,
\\ &
\Pi^{\progB}_{t} := \sum_{R\in\rec^t} \Big(\sum_{z\in[n]\setminus\inrec{R}}|f_z\>\<f_z|-|\psi_{\inrec{R}}\>\<\psi_{\inrec{R}}|\Big) \otimes |R\>\<R|,
\\ &
\Pi^{\progA}_{t} := \sum_{R\in\rec^t} |\psi_{\inrec{R}}\>\<\psi_{\inrec{R}}| \otimes |R\>\<R|
\end{align*}
on registers $\regFunc\regR$, which act as the identity on the algorithm registers $\regA$. We denote the spaces corresponding to these projectors, that is, their images, by $\cH^\progC_t,\cH^\progB_t,\cH^\progA_t$, respectively. We might drop the subscript $t$ when it is clear from the context.

Similarly to~\cite{hamoudi:2022:tradeoffs}, we define the \emph{progress measure} as
\[
\Psi_t:=\|\Pi^\progC_{t}|\phi_t\>\|^2 + 3 \|\Pi^\progB_{t}|\phi_t\>\|^2,
\]
which is essentially an upper bound on the provisional success probability of the algorithm after $t$ oracle calls (see the second claim of Lemma~\ref{lem:progEvol}).
We elaborate on the choice for the scalar $3$ in front of $ \|\Pi^\progB_{t}|\phi_t\>\|^2$ in Remark~\ref{rmk:FourInProg}.

\subsection{Main lemma on the progress measure}

Here we state Lemma~\ref{lem:progEvol} and show how it easily leads to Theorem~\ref{thm:main}. Then, the rest of the paper is devoted to proving the lemma, its first two claims being relatively easy to show.

\begin{lem}
\label{lem:progEvol}
We have
\begin{subequations}
\begin{align}
& \Psi_0 = 0, \\
& q_{succ} \le \Psi_\tau +  \frac2{n-\tau}, \\
& \Psi_{t+1}-\Psi_{t} \le \frac{48}{p(n-t-1)}.
\end{align}
\end{subequations}
\end{lem}

\begin{proof}[Proof of Theorem~\ref{thm:main} given Lemma~\ref{lem:progEvol}]
From Lemma~\ref{lem:progEvol}, we see that the success probability $q_{succ}$ is at most
\[
\frac{2}{n-\tau} + \frac{48}{p}\sum_{t=1}^\tau \frac{1}{n-t} 
\le \frac{2}{n-\tau} + \frac{48\tau}{p(n-\tau)} 
= \frac{2p+48\tau}{p(n-\tau)}
\]
Since we want the success probability to be at least $1-\epsilon$, we thus get
\[
\tau
\ge \frac{p(n(1-\epsilon)-2)}{48+p(1-\epsilon)}
\ge \frac{pn(1-\epsilon)}{49}-1.
\qedhere
\]
\end{proof}

Now it is left to prove Lemma~\ref{lem:progEvol}.
Its first claim is trivial, because initially the record is empty and for the empty record $R$ we have $\inrec{R}=\emptyset$ and $|\psi_{\emptyset}\>=|\unif\>$, which is the initial state of the truth register.

 Let us now prove the second claim of the lemma, while the proof of the final claim is much more involved, and we leave it to Sections~\ref{sec:transit} and \ref{sec:boundProg}.

\bigskip

\begin{proof}[Proof of the second claim of Lemma~\ref{lem:progEvol}]
To prove the second claim of Lemma~\ref{lem:progEvol}, recall 
$\Pi_{succ} = \sum_{z\in[n]}|f_z\>\<f_z|\otimes|z\>\<z|$.
It can be easily seen that $\Pi_{succ}$ commutes with both $\Pi^\progC_\tau$ and $\Pi^\progB_\tau+\Pi^\progA_\tau$, all being diagonal in the computational basis (also see Claim~\ref{clm:nonAlter} below). Hence
\begin{align*}
q_{succ} = \, & \|\Pi_{succ}\Pi^\progC_\tau|\phi_\tau\>\|^2 + \|\Pi_{succ}(\Pi^\progB_\tau+\Pi^\progA_\tau)|\phi_\tau\>\|^2
\\ \le \, & \|\Pi^\progC_\tau|\phi_\tau\>\|^2 + 
\big(
\|\Pi^\progB_\tau|\phi_\tau\>\|
+ \|\Pi_{succ}\Pi^\progA_\tau\|
\big)^2
\\ \le \, & \|\Pi^\progC_\tau|\phi_\tau\>\|^2 + 
2\|\Pi^\progB_\tau|\phi_\tau\>\|^2
+ 2\|\Pi_{succ}\Pi^\progA_\tau\|^2
\\ \le \, & \Psi_\tau + 2\|\Pi_{succ}\Pi^\progA_\tau\|^2,
\end{align*}
where we have used $\|\Pi_{succ}\|=1$ and $\||\phi_\tau\>\|=1$ for the first inequality.
To conclude, we have
\[
\|\Pi_{succ}\Pi^\progA_\tau\| = \max_{\substack{R\in\rec^\tau\\z\in[n]}}
\||f_z\>\<f_z|\psi_{\inrec{R}}\>\<\psi_{\inrec{R}}|\| = 1/\sqrt{n-\tau}.
\qedhere
\]
\end{proof}

\section{Transitions Among Progress-defining Subspaces}
\label{sec:transit}

Here we first decompose $\cH^\progB_{t}$ as a direct sum $\cH^{\progB,\actv}_{t}\oplus\cH^{\progB,\pasv}_{t}$, and then we provide various claims that will serve as basis of proving Lemma~\ref{lem:progEvol}.

\subsection{Active and passive intermediate subspaces}
\label{sec:ActAndPas}

It is important to note that the progress measure $\Psi_t$ is not affected by any operations on the algorithm registers alone, in particular, unitaries $U_t$. That is because $\Pi^\progC_{t}$, $\Pi^\progB_{t}$, $\Pi^\progA_{t}$ all act as the identity on the algorithm registers. However, when analyzing how $\Psi_t$ evolves under oracle calls, it is useful to decompose $\Pi^\progB_{t}$ further, this decomposition involving query input register as well.

 For $R\in\rec^*$ and $x\in [n]\setminus\inrec{R}$, let us define the approximation of $|f_x\>$ with respect to $R$ as the unit vector
\begin{align*}
|\tilde{f}_{x,R}\> := \,&
 \sqrt{\frac{n_R-1}{n_R}}|f_x\> - \frac1{\sqrt{n_R(n_R-1)}}\sum_{\substack{x'\in[n]\setminus\inrec{R}\\x'\ne x}}|f_{x'}\>,
 \end{align*}
 which we can also rewrite as 
 \begin{align*}
|\tilde{f}_{x,R}\>  = \, & \sqrt{\frac{n_R-1}{n_R}}|f_x\> - |\psi_{\inrec{(R\app x)}}\>
 \frac1{\sqrt{n_R}}
=
  \frac{\sqrt{n_R}|f_x\>-|\psi_{\inrec{R}}\>}{\sqrt{n_R-1}}.
\end{align*}
 We note that $\<\psi_{\inrec{R}}|\tilde{f}_{x,R}\>=0$, $\<\psi_{\inrec{(R\app x)}}|f_{x}\>=0$, and
 \begin{equation}
 \label{eq:fxANDfxapprox}
|\psi_{\inrec{R}}\>\<\psi_{\inrec{R}}|+|\tilde{f}_{x,R}\>\<\tilde{f}_{x,R}|
=
|\psi_{\inrec{(R\app x)}}\>\<\psi_{\inrec{(R\app x)}}|+|f_x\>\<f_x|.
 \end{equation}
 Also note that $|\tilde{f}_{x,R}\>$ is the same for all $R$ with the same $\inrec{R}$.

For every $t$, let us decompose
\[
\Pi^\progB_{\regFunc\regR}\otimes I_\regQi = \Pi^{\progB,\actv}_{\regFunc\regR\regQi}+\Pi^{\progB,\pasv}_{\regFunc\regR\regQi},
\]
where
\begin{align*}
& \Pi^{\progB,\actv}_{t} := \sum_{R\in\rec^t} \sum_{x\in [n]\setminus\inrec{R}}
|\tilde{f}_{x,R},R,x\>\<\tilde{f}_{x,R},R,x|, \\
& \Pi^{\progB,\pasv}_{t} := \sum_{R\in\rec^t}
\Big[
\sum_{x\in [n]\setminus\inrec{R}}
\Big(\sum_{z\in[n]\setminus\inrec{(R\app x)}}|f_z\>\<f_z|
-|\psi_{\inrec{(R\app x)}}\>\<\psi_{\inrec{(R\app x)}}|\Big)
\otimes |R,x\>\<R,x|
\\ & \qquad\qquad\qquad\qquad
 + \sum_{x\in\inrec{R}}
\Big(\sum_{z\in[n]\setminus\inrec{R}}|f_z\>\<f_z|
-|\psi_{\inrec{R}}\>\<\psi_{\inrec{R}}|\Big)
\otimes |R,x\>\<R,x|\Big]
\\ & \qquad\quad
= \sum_{R\in\rec^t}
\sum_{x\in [n]}
\Big(\sum_{z\in[n]\setminus\inrec{(R\app x)}}|f_z\>\<f_z|
-|\psi_{\inrec{(R\app x)}}\>\<\psi_{\inrec{(R\app x)}}|\Big)
\otimes |R,x\>\<R,x|.
\end{align*}
It is easy to see that $ \Pi^{\progB,\actv}_{t} $ and $ \Pi^{\progB,\pasv}_{t} $ are orthogonal projectors and, using (\ref{eq:fxANDfxapprox}), that they indeed sum up to $ \Pi^{\progB}_{t} $.
We call the subspace $\cH^{\progB,\actv}_{t}$ corresponding to $\Pi^{\progB,\actv}_{t}$ the \emph{active subspace} and  the subspace $\cH^{\progB,\pasv}_{t}$ corresponding to $\Pi^{\progB,\pasv}_{t}$ the \emph{passive subspace}.

\subsection{Intuition behind the bound}
\label{sec:intuit}

Before Sections~\ref{sec:non-alterability}--\ref{sec:escaping-no-progress}, where we establish rigorous claims on how the classical and the quantum oracle calls can change the overlap of the current state of the memory on spaces $\cH^{\progA}$, $\cH^{\progB,\actv}$, $\cH^{\progB,\pasv}$, and  $\cH^{\progC}$---the claims that we use to prove Lemma~\ref{lem:progEvol} in Section~\ref{sec:boundProg}---let us try to obtain some intuition behind the proof of the bound.
This section, Section~\ref{sec:intuit}, is rather vague and expresses author's intuition behind the result, the intuition that was partially formed before the completion of the result.
This intuition will hopefully illuminate the motivation behind the upcoming claims, which we tie together in~Section~\ref{sec:boundProg}.

\bigskip

Up to a negligible additive term, the progress measure $\Psi_t$ is an upper bound on the success probability of the computation (the second claim of Lemma~\ref{lem:progEvol}), and, by inspecting the definition of the progress measure, we can see that, in order for the computation to succeed, one has to gradually transfer the state of the system from $\cH^{\progA}$ to $\cH^{\progB} \oplus \cH^{\progC}$. More quantitively, one has to reduce the \emph{probability (weight)} $\|\Pi^{\progA}|\phi\>\|^2$ on $\cH^{\progA}$ or, equivalently, the \emph{amplitude} $\|\Pi^{\progA}|\phi\>\|$ on $\cH^{\progA}$, where $|\phi\>$ is the state of the overall system.

The probability weights on $\cH^{\progA}$, $\cH^{\progB,\actv}$, $\cH^{\progB,\pasv}$, $\cH^{\progC}$ sum up to $1$, and thus we could easily formalize the concept of transferring the probability weight from one subspace to another. However, to understand the dynamics of the computation and how the overall state gets transferred from one subspace to another, the probability weight might not be the most instructive concept to consider.
 
As an example, consider the vanilla Grover's algorithm in the noiseless setting (i.e., $p=0$). In this case, only the spaces $\cH^{\progA}$ and $\cH^{\progB,\actv}$ are utilized, and we might be interested in inspecting $\|\Pi^{\progB,\actv}_{t+1}|\phi_{t+1}\>\|^2-\|\Pi^{\progB,\actv}_{t}|\phi_{t}\>\|^2$, the probability weight transferred ``towards the success'' by a single oracle call. However, the value of this ``probability gain'' changes notably form one oracle call to another---it increases about linearly in $t$---but, when proving lower bounds, it is often easier to focus on quantities whose change is bounded by the same value, irrespective from the number of a given oracle call. Hence, it may be more useful to inspect the ``state transfer'' $\Pi^{\progB,\actv}_{t+1} O_Q\Pi^{\progA}_{t}|\phi_{t}\>$.%
\footnote{We could think of this state transfer as $\Pi^{\progB,\actv}_{t+1}|\phi_{t+1}\>-\Pi^{\progB,\actv}_{t}|\phi_{t}\>$, but it might be a bit misleading, because this latter quantity can be made artificially large, say, by introducing the global phase $-1$. Also, the notation $\Pi^{\progB,\actv}_{t+1} O_Q\Pi^{\progA}_{t}|\phi_{t}\>$ better expresses from which space to which space the state gets transferred.}
That is because Grover's algorithm maintains the norm of this state transfer about the same for every oracle call, $2/\sqrt{n}$, and, in fact, this norm can by shown%
\footnote{Essentially, following the proof of Claim~\ref{clm:OQonPiA}, but with the record always remaining empty.}
 to be at most $2/\sqrt{n}$ within any (noiseless) algorithm. In addition, for Grover's algorithm, all these state transfers are parallel and thus their norms add up, which is the reason why the success probability of Grover's algorithm scales quadratically in the total number of queries.

\bigskip

Now, let us again look at the noisy scenario, and consider the following ideas for an algorithm. Suppose we have run an instance of Grover's algorithm for some number of queries, we have been lucky and so far no error has occurred, and we have generated some partial success, meaning that, we have increased $\| \Pi^{\progB}|\phi\>\|$. Before our luck runs out, we may wish to safeguard this partial progress form the noise, so the question is: Can we do that?
Yes, we indeed can: we can copy the state from the query register $\regQ$, which is the only register affected by the noise, to somewhere within the workspace register $\regW$, and let it rest there, unaffected by noise.
Then, as the next step, we could again initialize the query registers in a uniform superposition $\frac1{\sqrt n}\sum_x{|x\>_\regQi|-\>_\regQo}$ and start running a fresh instance of Grover's algorithm. This way, each query will again yield the state transfer $\Pi^{\progB}_{t+1} O_p\Pi^{\progA}_{t}|\phi_{t}\>$ of norm about $2/\sqrt{n}$, while having the comfort of having some progress ``saved'' in the noiseless workspace $\regW$.

So, it might seem like we can get the best of the both worlds: we have accumulated and  preserved significant amount of success amplitude $\|\Pi^\progB|\phi\>\|$ in the workspace $\regW$ and we are also continuing to add to $\cH^{\progB}$ fresh state transfer of norm $2/\sqrt{n}$, so it may seem like the success probability, which is the square of the ``success amplitude'', might increase with an accelerating rate, just like in the noiseless Grover's algorithm. Such an approach may seem even more promising when we have flag bits that signal occurrences of the error.

However, we \emph{cannot} get the best of the both worlds: as it turns out, the state transfers freshly added to $\cH^\progB$ are orthogonal to the ``partial success state'' already saved in the workspace. And, hence, the norms of these two states do not add up (as if they were parallel); instead, their squares do, so the increase in the success probability is of a much smaller order.

This is essentially the reasoning for decomposing $\cH^{\progB}$ as $\cH^{\progB,\actv}\oplus\cH^{\progB,\pasv}$. The space $\cH^{\progB,\actv}$ is where we can ``actively'' increase the progress (by transferring some probability weight from $\cH^{\progA})$ and the space which is affected by the noise, while its orthogonal complement within $\cH^{\progB}$ is where we can passively store some so-far-acquired progress, it being safe from the noise.


\definecolor{GreenOq}{cmyk}{0.7, 0, 0.9, 0.3}
\definecolor{RedOc}{rgb}{1,0,0}
  
  \def\hgap{0.1}
  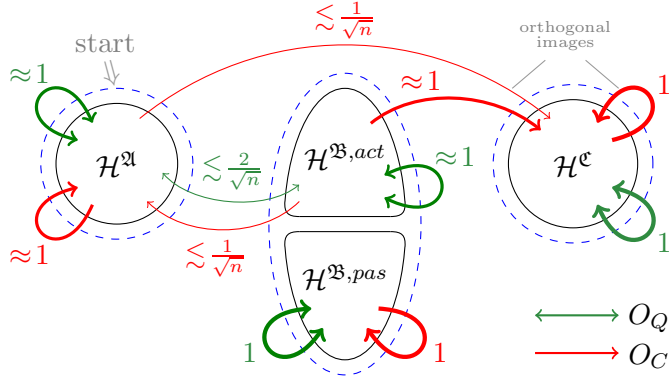
\begin{figure}[!h]
\centering
\begin{tikzpicture}

     \coordinate (A) at (-3,0.8);
     \draw (A) circle (1-2*\hgap);
     \draw [dashed, blue] (A) circle (1);
     \node [] at (A) {$\cH^{\progA}$};

  \coordinate (A)   at (3,0.8);
     \draw (A) circle (1-1.5*\hgap);
     \draw [dashed, blue] (A) circle (1+0.5*\hgap);
     \node [] at (A) {$\cH^{\progC}$};

       \draw [dashed, blue] (0,0) ellipse (1 and 2);
       \draw [] (0,0) ellipse ({1-2*\hgap} and {2-2*\hgap} );
       \draw [] ({-1+4*\hgap},\hgap) -- ({1-4*\hgap},\hgap);
       \draw [] ({-1+4*\hgap},-\hgap) -- ({1-4*\hgap},-\hgap);
       
       \filldraw [white] ({-1+\hgap},{-3*\hgap+0.01}) rectangle ({-1+3*\hgap},{3*\hgap-0.01});
       \filldraw [white] ({1-3*\hgap},{-3*\hgap+0.01}) rectangle ({1-\hgap},{3*\hgap-0.01});

       \draw [] ({-1+4*\hgap} ,\hgap) .. controls 
       ({-1+2.5*\hgap} ,\hgap) and 
       ({-1+2*\hgap} ,1.5*\hgap) ..  
       ({-1+2*\hgap+0.01} ,{3*\hgap});

       \draw [] ({1-4*\hgap} ,\hgap) .. controls 
       ({1-2.5*\hgap} ,\hgap) and 
       ({1-2*\hgap} ,1.5*\hgap) ..  
       ({1-2*\hgap-0.01} ,{3*\hgap});
       
       \draw [] ({-1+4*\hgap} ,-\hgap) .. controls 
       ({-1+2.5*\hgap} ,-\hgap) and 
       ({-1+2*\hgap} ,-1.5*\hgap) ..  
       ({-1+2*\hgap+0.01} ,{-3*\hgap});

       \draw [] ({1-4*\hgap} ,-\hgap) .. controls 
       ({1-2.5*\hgap} ,-\hgap) and 
       ({1-2*\hgap} ,-1.5*\hgap) ..  
       ({1-2*\hgap-0.01} ,{-3*\hgap});
       \node [shift={(0.,0.)}] at (0,.9) {$\cH^{\progB,\actv}$};
       \node [shift={(0.,0.)}] at (0,-.7) {$\cH^{\progB,\pasv}$};
 
\draw [ultra thick, green!50!black, <->, shift={(-0.8,-1.5)}, rotate = -60]
 (0.15,0.5) .. controls (0.65,-0.5) and (-0.65,-0.5) .. (-0.15,0.5);
\draw [ultra thick, red, ->, shift={(0.8,-1.5)},rotate = 60]
 (0.15,0.5) .. controls (0.65,-0.5) and (-0.65,-0.5) .. (-0.15,0.5);
 \node [GreenOq, shift={(-0.45,-0.2)}] at (-0.8,-1.5){$1$};
\node [RedOc, shift={(0.45,-0.2)}] at (0.8,-1.5) {$1$};
\draw [ultra thick, RedOc, <-, shift={(3.8,1.55)}, rotate = 130]
 (0.15,0.5) .. controls (0.65,-0.5) and (-0.65,-0.5) .. (-0.15,0.5);
\draw [ultra thick, GreenOq, <->, shift={(3.8,0.05)},rotate = 50]
 (0.15,0.5) .. controls (0.65,-0.5) and (-0.65,-0.5) .. (-0.15,0.5);
\node [RedOc, shift={(0.4,0.3)}] at (3.8,1.55){$1$};
\node [GreenOq, shift={(0.4,-0.3)}] at (3.8,0.05) {$1$};
\draw [very thick, green!50!black, <->, shift={(-3.8,1.55)}, rotate = -130]
 (0.15,0.5) .. controls (0.65,-0.5) and (-0.65,-0.5) .. (-0.15,0.5);
\draw [very thick, red, ->, shift={(-3.8,0.05)},rotate = -50]
 (0.15,0.5) .. controls (0.65,-0.5) and (-0.65,-0.5) .. (-0.15,0.5); 
 \node [GreenOq, shift={(-0.35,0.45)}] at (-3.8,1.55){$\approx\! 1$};
\node [RedOc, shift={(-0.35,-0.45)}] at (-3.8,0.05) {$\approx\! 1$};
\draw [very thick, green!50!black, <->, shift={(1.02,0.5)},rotate = 90]
 (0.15,0.5) .. controls (0.65,-0.5) and (-0.65,-0.5) .. (-0.15,0.5);
 \node [GreenOq] at (1.45,0.9){$\approx\! 1$};

 \draw [red, ->, shift={(0,1.4)}]
 (-2.7,0) .. controls  (-1.1,1.3) and  (1.1,1.3) ..  (2.7,0);
  \node [RedOc] at (0,2.7){\small $\lesssim\! \tfrac{1}{\sqrt{n}}$};
\draw [GreenOq, <->, shift={(-1.5,0.6)}, rotate=-5]
 (-.9,0) .. controls  (-0.5,-0.4) and  (0.5,-0.4) ..  (.9,0);
 \draw [RedOc, <-, shift={(-1.6,0.3)}]
 (-1,0) .. controls  (-0.5,-0.5) and  (0.5,-0.5) ..  (1,0);
    \node [GreenOq] at (-1.48,0.7){\small $\lesssim\! \tfrac{2}{\sqrt{n}}$};
   \node [RedOc] at (-1.7,-0.45){\small $\lesssim\! \tfrac{1}{\sqrt{n}}$};
 \draw [very thick, red, ->, shift={(1.535,1.1)}, rotate = 8]
 (-1.15,0.4) .. controls  (-0.5,0.8) and  (0.5,0.6) ..  (1.05,0);
 \node [RedOc, shift={(0.4,0.2)}] at (0.6,1.7){$\approx\! 1$};

\node [gray!80!white, shift={(-3.1,2.)}, rotate=-75] at (0,0) {$\Rightarrow$};
\node [gray, shift={(-3.1,2.)}] at (-0.05,0.4) {start};

\node [gray] at (2.9,2.6) {\tiny orthogonal};
\node [gray] (Nim) at (2.9,2.4) {\tiny images};

\draw [gray!80!white] (Nim) -- (2.2,1.8);
\draw [gray!80!white] (Nim) -- (3.6,1.8);


\draw [GreenOq, <->, shift={(4,-1.25)}, line width=0.8pt] (-1.5,0.04) --  (-0.4,0.04);
\node [] at (4,-1.25) {$O_Q$};
\draw [RedOc, ->, shift={(4,-1.75)}, line width=0.8pt] (-1.5,0.04) --  (-0.4,0.04);
\node [] at (4,-1.75) {$O_C$};

\end{tikzpicture}
\captionsetup{font=small}
\captionsetup{width=0.9\textwidth}
\caption[my caption]{%
Possible state transfers between subspaces $\cH^{\progA},\cH^{\progB,\actv},\cH^{\progB,\pasv},\cH^{\progC}$ by oracle calls. More formally, the value by the arrow from a space $\cH^{\mathsf{label}_1}$ to a space $\cH^{\mathsf{label}_2}$ indicates the norm of $\Pi^{\mathsf{label}_2}_{t+1} O \Pi^{\mathsf{label}_1}_{t}$ where $t\ll n$ and $O$ is either $O_Q$ or $O_C$. 
Double-headed green arrows $\textcolor{GreenOq}{\longleftrightarrow}$ indicate the quantum oracle $O_Q$ and single-headed red arrows $\textcolor{RedOc}{\longrightarrow}$ the classical oracle $O_C$; the width of an arrow suggests the value of the norm. 
For quantum oracle calls, we have $\|\Pi^{\mathsf{label}_2}_{t+1} O_Q \Pi^{\mathsf{label}_1}_{t}\|=\|\Pi^{\mathsf{label}_1}_{t+1} O_Q \Pi^{\mathsf{label}_2}_{t}\|$ because, aside from introducing the flag $|\flagQ\>_{\regW^+}$ and the record $|\bot\>_{\regR_{t+1}}$, $O_Q$ is its own inverse.

\hspace{10pt}
The blue dashed lines enclose together spaces between which unitaries $U_t$ can make state transfers.
At the beginning of the computation, all the probability weight is on $\cH^{\progA}$. The images of $\Pi^{\progC}_{t+1} O_C \Pi^{\progC}_{t}$ and $\Pi^{\progC}_{t+1} O_C (\Pi^{\progB}_{t}+\Pi^{\progA}_{t})$ are orthogonal.}
\label{fig:InterspaceTransfers}
\end{figure}

We may ask:
What would happen if we didn't try to store away any previously acquired progress in the noise-unaffected subspace $\cH^{\progB,\pasv}$? On one hand side, we could have all the state transfers $\Pi^{\progB,\actv}_{t+1} O_p\Pi^{\progA}_{t}|\phi_{t}\>$ pointing in the same direction and thus their norms adding up.
On the other hand, a $p$-dependent fraction of the formerly acquired progress $\|\Pi^{\progB,\actv}_{t}|\phi_{t}\>\|$ would ``leak'' from the quantum-success subspace $\cH^{\progB,\actv}$ to the classical-success subspace $\cH^{\progC}$ (see Claim~\ref{clm:OQonAct} and Footnote~\ref{foot:AtoBact}). 
Ultimately, some kind of equilibrium would likely be reached, and the amount of the probability weight coming in from $\cH^{\progA}$ to $\cH^{\progB,\actv}$ would be equal to that leaking out from $\cH^{\progB,\actv}$ to $\cH^{\progC}$.

Here we are again talking about probability weight, because, whenever an oracle call (in particular, the classical component $O_C$ of $O_p$) transfers state from $\cH^{\progA}\oplus\cH^{\progB,\actv}$ to $\cH^{\progC}$, it transfers it to an orthogonal subspace of $\cH^{\progC}$. More formally, $\Pi^{\progC}_{t+1} O_p(\Pi^{\progA}_{t}+\Pi^{\progB,\actv}_{t})|\phi_{t}\>$ and $\Pi^{\progC}_{t+1} O_p\Pi^{\progC}_{t}|\phi_{t}\>$ are orthogonal. This is because all state transfers to $\cH^{\progC}$ effectively come with a time-stamp $t+1$ indicating when the dephasing on the marked element happened for the first time.


In Figure~\ref{fig:InterspaceTransfers}, we show how much ``norm'' gets transferred (or can get transferred) form one space to another by a single extended oracle call $O_p=\sqrt{1-p}O_Q+\sqrt{p}O_C$. The quantities displayed in the figure are due to Claims~\ref{clm:nonAlter}--\ref{clm:OConPiA}.

\subsection{Non-alterability of the record}
\label{sec:non-alterability}

Let us now establish various claims towards formalizing the intuitive argument made above.
To start with, observe that $\Pi^\progC$ commutes with $|f_z\>\<f_z|$ for every $z$, and thus so does $I-\Pi^\progC=\Pi^\progB+\Pi^\progA$. Now let us prove the following claim, which, informally speaking, states that a quantum oracle call cannot change whether a marked element is in the record, that a marked element in the record cannot be erased by a classical oracle call, and that states with orthogonal records remain orthogonal under oracle calls.

\begin{clm}
\label{clm:nonAlter}
We have $O_Q\Pi^{\progC}_t = \Pi^{\progC}_{t+1}O_Q$ and $O_C\Pi^{\progC}_t = \Pi^{\progC}_{t+1}O_C\Pi^{\progC}_t$, and the images of $\Pi^\progC_{t+1} O_C\Pi^\progC_t$ and $\Pi^\progC_{t+1} O_C(\Pi^\progB_t+\Pi^\progA_t)$ are orthogonal.
\end{clm}

\noindent
Note that we do not claim the equality of $\Pi^{\progC}_{t+1}O_C$ and $O_C\Pi^{\progC}_t$.

\begin{proof}
For the first claim,
since $O_Q$ appends $\bot$ to the record, yet $\inrec{R}=\inrec{(R\app\bot)}$, we have
\begin{align*}
O_Q\Pi^{\progC}_{t}
= &\, \sum_{R\in\rec^t} \sum_{z\in\inrec{R}} O_Q \big(|f_z\>\<f_z| \otimes |R\>\<R|\big)
\\ = &\, \sum_{R\in\rec^t} \sum_{z\in\inrec{R}} 
|f_z\>\<f_z| \otimes O_{f_z}\otimes|\flagQ,R\app\bot\>\<R|
\\ = &\, \sum_{\substack{R'\in\rec^{t+1}\\R'_{t+1}=\bot}} \sum_{z\in\inrec{R'}}  \big(|f_z\>\<f_z| \otimes |R'\>\<R'|\big) O_Q 
=
  \Pi^{\progC}_{t+1}O_Q,
\end{align*}
where we have used that $\<R'|O_Q=0$ whenever the last entry of the record $R'$ is not $\bot$.

For the second claim, because $|f_z\>\<f_z|$ commutes with $O_C$ and because the oracle call appends some symbol $r\in\rec$ to the record, we have
\begin{align*}
O_C\Pi^{\progC}_{t}
= &\, \sum_{R\in\rec^t} \sum_{z\in\inrec{R}} O_C\big(|f_z\>\<f_z| \otimes |R\>\<R|\big)
\\ = &\, \sum_{R\in\rec^t} \sum_{z\in\inrec{R}} \sum_{r\in\rec}
\big(|f_z\>\<f_z| \otimes |R\app r\>\<R\app r|\big) 
O_C\big(|f_z\>\<f_z| \otimes |R\>\<R|\big).
\end{align*}
For every $x\in[n]$, we clearly have $\inrec{(R\app x)}=\inrec{R}\cup\{x\}\supseteq\inrec{R}$.
Thus, for $R\in\rec^t$, $z\in\inrec{R}$, and $x\in[n]$, we have 
\[
\Pi^{\progC}_{t+1} \big(|f_z\>\<f_z| \otimes |R\app x\>\<R\app x|\big) = |f_z\>\<f_z| \otimes |R\app x\>\<R\app x|,
\]
which means that $O_C\Pi^{\progC}_{t}$, as expressed above, is unaffected when multiplied by $\Pi^{\progC}_{t+1}$ from the left.

For the final claim, if we look at the truth $\regFunc$ and the record $\regR$ registers of the image of $\Pi^{\progB}_{t}+\Pi^{\progA}_{t}$, it is spanned by vectors in form $|f_z,R\>$, where $R\in\rec^t$ and $z\in[n]\setminus\inrec{R}$. Hence, the image of $\Pi^{\progC}_{t+1} O_C (\Pi^{\progB}_{t}+\Pi^{\progA}_{t})$ restricted to those registers is spanned by vectors in form $|f_z,R\app z\>$ where $R\in\rec^t$ and $z\in[n]\setminus\inrec{R}$. 
On the other hand, the image of $\Pi^{\progC}_{t}$ restricted to $\regFunc\regR$ is spanned by vectors in form $|f_z,R\>$, where $R\in\rec^t$ and $z\in\inrec{R}$.  Hence, the image of $\Pi^{\progC}_{t+1} O_C \Pi^{\progC}_{t}$ restricted to $\regFunc\regR$ is spanned by vectors in form $|f_z,R\app x\>$ where $R\in\rec^t$,  $z\in\inrec{R}$, and $x\in[n]$. This concludes the proof.
\end{proof}

\subsection{The action of $O_Q$ and $O_C$ on the active and passive subspaces}

First we show that the passive subspace $ \cH^{\progB,\pasv}$ is invariant under both the quantum and the classical oracle calls.

\begin{clm}
\label{clm:OQonPas}
We have both $\cH^{\progB,\pasv}_{t+1} O_Q = O_Q \cH^{\progB,\pasv}_{t}$ and $\cH^{\progB,\pasv}_{t+1} O_C = O_C \cH^{\progB,\pasv}_{t}$, and thus $\cH^{\progB,\pasv}_{t+1} O_p = O_p \cH^{\progB,\pasv}_{t}$.
\end{clm}

\begin{proof}
Recall that $O_{f_x}=I_{2n}-2|x,-\>\<x,-|$ and consider
\[
\sum_{x\in[n]}|f_x\>\<f_x|\otimes O_{f_x} 
= I-2\sum_{x\in[n]}|f_x,x,-\>\<f_x,x,-|.
\]
Note that $\Pi^{\progB,\pasv}_{t}|f_x,x,R\>=0$ for all $x\in[n]$, all $t$, and all $R\in\rec^t$, irrespective of whether $x\in\inrec{R}$. Hence,
\[
\Big(
\sum_{x\in[n]}|f_x\>\<f_x|\otimes O_{f_x} 
\Big)
\Pi^{\progB,\pasv}_{t} = \Pi^{\progB,\pasv}_{t}
=
\Pi^{\progB,\pasv}_{t}
\Big(
\sum_{x\in[n]}|f_x\>\<f_x|\otimes O_{f_x} 
\Big).
\]
This means that both
\begin{align*}
O_Q \Pi^{\progB,\pasv}_{t} 
& =\Big(I_\regQi\otimes|\flagQ,\bot\> \Big)\Pi^{\progB,\pasv}_{t}
=\Pi^{\progB,\pasv}_{t}\otimes|\flagQ,\bot\> 
\\ & =\Pi^{\progB,\pasv}_{t+1} \Big(I_\regQi\otimes|\flagQ,\bot\> \Big)
= \Pi^{\progB,\pasv}_{t+1} O_Q,
\end{align*}
which already establishes the claim for the quantum oracle, 
and
\begin{align*}
O_C \Pi^{\progB,\pasv}_{t} & =\Big(\sum_{x\in[n]}|x\>\<x|\otimes|\flagC,x\> \Big)\Pi^{\progB,\pasv}_{t}
\\ & = 
\sum_{R\in\rec^t}
\sum_{x\in [n]}
\Big(\sum_{z\in[n]\setminus\inrec{(R\app x)}}|f_z\>\<f_z|
-|\psi_{\inrec{(R\app x)}}\>\<\psi_{\inrec{(R\app x)}}|\Big)
\\ & \hspace{230pt}
\otimes |\flagC\>\otimes|R\app x,x\>\<R,x|.
\end{align*}

Now let us consider $\Pi^{\progB,\pasv}_{t+1}O_C$.
We can write
\begin{align*}
\Pi^{\progB,\pasv}_{t+1} 
 & = \sum_{R\in\rec^t}\sum_{r\in\rec}
\sum_{x\in [n]}
\Big(\sum_{z\in[n]\setminus\inrec{(R\app rx)}}|f_z\>\<f_z|
-|\psi_{\inrec{(R\app rx)}}\>\<\psi_{\inrec{(R\app rx)}}|\Big)
\\ & \hspace{245pt}
\otimes |R\app r,x\>\<R\app r,x|.
\end{align*}
 Because $O_C$ adds some $x'\in[n]$ to the record, $O_C$ eliminates vectors with $r=\bot$ in the above expression, and we have
\begin{multline*}
\Pi^{\progB,\pasv}_{t+1} O_C =
\Big[ \sum_{R\in\rec^t}
\sum_{x,x'\in [n]}
\Big(\sum_{z\in[n]\setminus\inrec{(R\app x'x)}}|f_z\>\<f_z|
-|\psi_{\inrec{(R\app x'x)}}\>\<\psi_{\inrec{(R\app x'x)}}|\Big)
\\
\otimes |\flagC\>\<\flagC|
\otimes |R\app x',x\>\<R\app x',x| \Big] O_C.
\end{multline*}
Because the content of the register $\regQi$ determines what $O_C$ adds to the record, we must have $x'=x$ above. Thus, because of $\inrec{(R\app x x)}=\inrec{(R\app x)}$, we get
\begin{align*}
\Pi^{\progB,\pasv}_{t+1} O_C & =
\Big[ \sum_{R\in\rec^t}
\sum_{x\in [n]}
\Big(\sum_{z\in[n]\setminus\inrec{(R\app x)}}|f_z\>\<f_z|
-|\psi_{\inrec{(R\app x)}}\>\<\psi_{\inrec{(R\app x)}}|\Big)
\\ & \hspace{180pt}
\otimes |\flagC\>\<\flagC|
\otimes |R\app x,x\>\<R\app x,x| \Big] O_C
\\ & =
\sum_{R\in\rec^t} \sum_{x\in [n]}
\Big(\sum_{z\in[n]\setminus\inrec{(R\app x)}}|f_z\>\<f_z|
-|\psi_{\inrec{(R\app x)}}\>\<\psi_{\inrec{(R\app x)}}|\Big)
\\ & \hspace{180pt}
\otimes|\flagC\>
\otimes |R\app x,x\>\<R,x| . \qedhere
\end{align*}
\end{proof}

While we do not need to claim anything regarding what happens if the quantum oracle call acts on the active subspace, $\cH^{\progB,\actv}$, we will need to use the following fact which says that the classical oracle completely ``removes'' everything from the active subspace.%
\footnote{The action of $Q_Q$ on $\cH^{\progB,\actv}$ illustrated in Figure~\ref{fig:InterspaceTransfers} comes from the observation that $Q_Q$ is essentially its own inverse.}

\begin{clm}
\label{clm:OQonAct}
We have $O_C\Pi^{\progB,\actv}_t=(\Pi^\progC_{t+1}+\Pi^\progA_{t+1})O_C\Pi^{\progB,\actv}_t$.
\end{clm}

\begin{proof}
The space $\cH^{\progB,\actv}_t$ is spanned by vectors in form $|\tilde{f}_{x,R},x,y,R\>$ (here we ignore the content of workspace registers),  where $R\in\rec^t$, $x \in [n]\setminus\inrec{R}$, and $y\in\{0,1\}$. We have
\begin{align*}
&
O_C |\tilde{f}_{x,R},x,y,R\>
\\
& \qquad = \sqrt{\frac{n_R-1}{n_R}}O_C|f_x,x,y,R\> 
- \frac1{\sqrt{n_R(n_R-1)}}\sum_{\substack{x'\in[n]\setminus\inrec{R}\\x'\ne x}}O_C|f_{x'},x,y,R\>
\\ 
& \qquad =  \sqrt{\frac{n_R-1}{n_R}}|f_x,x,\neg y,\flagC,R\app x\> 
- \frac1{\sqrt{n_R(n_R-1)}}\sum_{\substack{x'\in[n]\setminus\inrec{R}\\x'\ne x}}|f_{x'},x,y,\flagC,R\app x\>
\\ 
& \qquad =  \sqrt{\frac{n_R-1}{n_R}}|f_x,x,\neg y,\flagC,R\app x\> 
- \frac1{\sqrt{n_R}} |\psi_{\inrec{(R\app x)}}\>|x,y,\flagC,R\app x\>,
\end{align*}
where $\neg y:=y\oplus 1$.
The former vector is in $\cH^\progC_{t+1}$, while the latter is in $\cH^\progA_{t+1}$.%
\footnote{By slightly extending the proof, one can see that $\|\Pi^\progA_{t+1}O_C\Pi^{\progB,\actv}_t\|\le 1\sqrt{n-t}$, which, while unnecessary for our proof, is still illustrated in Figure~\ref{fig:InterspaceTransfers}.\label{foot:AtoBact}}
 \end{proof}

\subsection{Escaping the no-progress subspace}
\label{sec:escaping-no-progress}

Now let us show that, informally speaking, the oracle calls acting on states in $\cH^\progA$ can move them only to specific subspaces and can do it only partially, and we place limits on how much can be moved.

\begin{clm}
\label{clm:OQonPiA}
We have $\Pi^\progB_{t+1} O_Q\Pi^\progA_t = \Pi^{\progB,\actv}_{t+1} O_Q\Pi^\progA_t $, and its norm is 
$2\frac{\sqrt{n-t-1}}{n-t}$.
\end{clm}

\begin{proof}
Because $\Pi^{\progB}_{t+1} = \Pi^{\progB,\actv}_{t+1} + \Pi^{\progB,\pasv}_{t+1} $, the claimed operator equality follows directly from Claim~\ref{clm:OQonPas}.
We can see that
\begin{align*}
\Pi^{\progB,\actv}_{t+1} O_Q\Pi^\progA_t
= \, & 
\bigg(
\sum_{R'\in\rec^{t+1}} \sum_{x\in [n]\setminus\inrec{R'}}
|\tilde{f}_{x,R'},x\>\<\tilde{f}_{x,R'},x| \otimes |R'\>\<R'|
\bigg)
\\ & \quad\cdot
\bigg(\Big(I_{\regFunc\regQ}-2\sum_{x\in[n]} |f_x,x,-\>\<f_x,x,-|\Big)\otimes|\flagQ,\bot\>\bigg)
\\ & \quad\cdot
\bigg(\sum_{R\in\rec^t} |\psi_{\inrec{R}}\>\<\psi_{\inrec{R}}| \otimes |R\>\<R|\bigg)
\\ %
= \, & 
-2 \sum_{R\in\rec^{t}} \sum_{x\in[n]\setminus\inrec{R}} 
\Big(
|\tilde{f}_{x,R}\>\<\tilde{f}_{x,R}|
|f_x\>\<f_x| |\psi_{\inrec{R}}\>\<\psi_{\inrec{R}}|
\Big)
\\ & \hspace{75pt}\otimes|x,-\>\<x,-|\otimes|\flagQ\>\otimes |R\app\bot\>\<R|,
\end{align*}
where we have used $|\tilde{f}_{x,R\app \bot}\>=|\tilde{f}_{x,R}\>$.
In the sum above, the terms corresponding to distinct $x$ and $R$ are orthogonal, therefore 
\begin{align*}
\big\|\Pi^{\progB,\actv}_{t+1} O_Q\Pi^\progA_t\big\|
\, & =  
2 \max_{\substack{R\in\rec^t\\x\notin \inrec{R}}}  
\big|
\<\tilde{f}_{x,R}|f_x\>
\<f_x|\psi_{\inrec{R}}\>
\big|
\\ &
=
2 \max_{R\in\rec^t} \sqrt{\frac{n_R-1}{n_R}} \frac1{\sqrt{n_R}}
=2\frac{\sqrt{n-t-1}}{n-t}. \qedhere
\end{align*}
\end{proof}

\begin{clm}
\label{clm:OConPiA}
We have $\Pi^\progB_{t+1} O_C\Pi^\progA_t=0$ and $\|\Pi^\progC_{t+1} O_C\Pi^\progA_t\|^2=1/(n-t)$.
\end{clm}

\begin{proof}
We can write $O_C\Pi^\progA_t$ as the summation 
 \begin{align*}
O_C\Pi^\progA_t = \, & 
\Big(\sum_{\substack{x,z\in[n]\\y\in\{0,1\}}}  |f_z,x,y\oplus\delta_{x,z}\>\<f_z,x,y| \otimes|x\>\otimes|\flagC\>\Big)
\Big(
\sum_{R\in\rec^t} |\psi_{\inrec{R}}\>\<\psi_{\inrec{R}}| \otimes |R\>\<R|
\Big)
\\ = \, & M' + M''
\end{align*}
where $M'$ corresponds to all terms in the summation such that $z\ne x$, and $M''$ to those with $z=x$.
For both, we use (\ref{eq:fxPsiR}) to evaluate $\< f_x|\psi_{\inrec{R}}\>$, and let
$I_\regQo=|0\>\<0|+|1\>\<1|$
and 
$X_\regQo=|1\>\<0|+|0\>\<1|$.
We have
 \begin{align*}
M'
= \, &
\sum_{R\in\rec^t} 
\sum_{z\in[n]\setminus\inrec{R}}
\sum_{\substack{x\in[n]\\x\ne z}}
\frac{|f_z\>\<\psi_{\inrec{R}}|}{\sqrt{n_R}}\otimes |R\app x\>\<R| 
\otimes|x\>\<x| \otimes I_\regQo \otimes|\flagC\>
\\ = \, &
 \sum_{R\in\rec^t} 
\sum_{x\in[n]}
\bigg(\sum_{\substack{z\in[n]\setminus\inrec{R}\\z\ne x}}
\frac{|f_z\>}{\sqrt{n_R}}
|R\app x\>\bigg)
\<\psi_{\inrec{R}}|\<R|
\otimes|x\>\<x| \otimes I_\regQo \otimes |\flagC\>
\\ = \, &
 \sum_{R\in\rec^t} 
\sum_{x\in[n]}
\bigg(\sqrt{\frac{n_{R\app x}}{n_R}}
 |\psi_{\inrec{(R\app x)}}\>
|R\app x\>\bigg)
\<\psi_{\inrec{R}}|\<R|
\otimes|x\>\<x| \otimes I_\regQo \otimes |\flagC\>,
\end{align*}
whose image is in $\cH^\progA_{t+1}$,
and we have
 \begin{align*}
M''
= \, &
\sum_{R\in\rec^t} 
\sum_{z\in[n]\setminus\inrec{R}}
\frac{|f_z\>\<\psi_{\inrec{R}}|}{\sqrt{n_R}} \otimes |R\app z\>\<R| 
\otimes|z\>\<z| \otimes X_\regQo \otimes|\flagC\>,
\end{align*}
whose image is in $\cH^\progC_{t+1}$. Since the terms corresponding to distinct $z$ and $R$ are orthogonal, $\|M''\|=\max_{R\in\rec^t} 1/\sqrt{n_R}=1/\sqrt{n-t}$.
\end{proof}

\section{Bounding Increase in  Progress}
\label{sec:boundProg}

We now combine Claims \ref{clm:nonAlter}--\ref{clm:OConPiA} to show the following lemma, which bounds the size of overlaps on spaces $\cH^\progB_{t+1}$ and $\cH^\progC_{t+1}$ after the quantum and the classical oracle calls, given the size of overlaps on $\cH^{\progB,\actv}_{t}$, $\cH^{\progB,\pasv}_{t}$, and $\cH^{\progC}_{t}$ before those oracle calls. We then use this lemma to conclude the proof of Lemma~\ref{lem:progEvol} and, in turn, that of Theorem~\ref{thm:main}.

\begin{lem}
\label{lem:PiBCOQC}
For $|\phi_t\>$ the state of the whole system just before $(t+1)$-th oracle call, we have
\begin{subequations}
\begin{align}
&
\|\Pi^\progB_{t+1} O_Q |\phi_t\>\|^2
 \le \bigg(\| \Pi^{\progB,\actv}_t|\phi_t\> \|
+ \frac2{\sqrt{n-t-1}}
\bigg)^2
  + \|\Pi^{\progB,\pasv}_t|\phi_t\>\|^2,
  \label{eq:PiBOQbound}
\\ \label{eq:PiBOCbound} &
\|\Pi^\progB_{t+1} O_C|\phi_t\>\|^2
 = \|\Pi^{\progB,\pasv}_t|\phi_t\>\|^2 ,
\\ \label{eq:PiCOQbound} &
\|\Pi^\progC_{t+1} O_Q |\phi_t\>\|^2   =  \|\Pi^\progC_t |\phi_t\>\|^2,
\\  &
\|\Pi^\progC_{t+1} O_C|\phi_t\>\|^2  
 \le \|\Pi^\progC_t |\phi_t\>\|^2 
+ 2\|\Pi^{\progB,\actv}_t|\phi_t\>\|^2
+ \frac 2{n-t-1}.
\label{eq:PiCOCbound}
\end{align}
\end{subequations}
\end{lem}

\begin{proof}
Towards (\ref{eq:PiBOQbound}), first note that $\Pi^\progB_{t+1} O_Q\Pi^\progC_t=0$ due to Claim~\ref{clm:nonAlter}. Then, by Claim~\ref{clm:OQonPas}, we have
\[
\Pi^\progB_{t+1} O_Q
  = \Pi^{\progB,\actv}_{t+1} O_Q \Pi^{\progB,\actv}_t
  + \Pi^{\progB,\pasv}_{t+1} O_Q \Pi^{\progB,\pasv}_t
  + \Pi^{\progB,\actv}_{t+1} O_Q\Pi^\progA_t.
\]
Because $\Pi^{\progB,\actv}_{t+1}$ and $\Pi^{\progB,\pasv}_{t+1}$ have orthogonal images,
we therefore have
\begin{align*}
\|\Pi^\progB_{t+1} O_Q|\phi_t\>\|^2
  = & \| \Pi^{\progB,\actv}_{t+1} O_Q \Pi^{\progB,\actv}_t|\phi_t\> + \Pi^{\progB,\actv}_{t+1} O_Q\Pi^\progA_t |\phi_t\>\|^2
+ \|\Pi^{\progB,\pasv}_{t+1} O_Q \Pi^{\progB,\pasv}_t|\phi_t\>\|^2
\\
\le & \Big( \| \Pi^{\progB,\actv}_t|\phi_t\> \| + \|\Pi^{\progB,\actv}_{t+1} O_Q\Pi^\progA_t |\phi_t\>\|\Big)^2
+ \|\Pi^{\progB,\pasv}_t|\phi_t\>\|^2,
\end{align*}
and the inequality (\ref{eq:PiBOQbound}) follows due to $\|\Pi^{\progB,\actv}_{t+1} O_Q\Pi^\progA_t \|<\frac2{\sqrt{n-t-1}}$, as given by Claim~\ref{clm:OQonPiA}.

Towards (\ref{eq:PiBOCbound}), we have $\Pi^\progB_{t+1} O_C\Pi^\progC_t=0$ by Claim~\ref{clm:nonAlter}, $\Pi^\progB_{t+1} O_C\Pi^{\progB,\actv}_t=0$ by Claim~\ref{clm:OQonAct}, and $\Pi^\progB_{t+1} O_C\Pi^\progA_t=0$ by Claim~\ref{clm:OConPiA}. Hence,
$\Pi^\progB_{t+1} O_C = \Pi^\progB_{t+1} O_C \Pi^{\progB,\pasv}_t$,
and the equality (\ref{eq:PiBOCbound}) follows from Claim~\ref{clm:OQonPas}.

The equality (\ref{eq:PiCOQbound}) follows immediately from Claim~\ref{clm:nonAlter}.
Towards (\ref{eq:PiCOCbound}),
recall Claim~\ref{clm:nonAlter} stating that $\Pi^\progC_{t+1} O_C\Pi^\progC_t=O_C\Pi^\progC_t$ and that its image is orthogonal to that of $\Pi^\progC_{t+1} O_C(\Pi^\progB_t+\Pi^\progA_t)$.
Since $\Pi^\progC_{t+1} O_C \Pi^{\progB,\pasv}_t = 0$ by Claim~\ref{clm:OQonPas},  we have
\begin{align}
\notag
\|\Pi^\progC_{t+1} O_C|\phi_t\>\|^2  
& = \|\Pi^\progC_t |\phi_t\>\|^2 + \|\Pi^\progC_{t+1} O_C(\Pi^{\progB,\actv}_t + \Pi^\progA_t)|\phi_t\>\|^2
\\ 
& \le \|\Pi^\progC_t |\phi_t\>\|^2 
+ 2\|\Pi^{\progB,\actv}_t|\phi_t\>\|^2
+ 2\underbrace{\|\Pi^\progC_{t+1} O_C\Pi^\progA_t\|^2}_{=\frac1{n-t}<\frac1{n-t-1}},
\label{eq:sansansan}
\end{align}
where the equality under the brace is due to Claim~\ref{clm:OConPiA}.
\end{proof}

\begin{rmk}
\label{rmk:FourInProg}
Consider the last two norms on the left hand side of (\ref{eq:sansansan}), namely, $N_{\progB,\actv}:=\|\Pi^{\progB,\actv}_t|\phi_t\>\|$ and $N_{\progC\leftarrow\progA}:=
\|\Pi^\progC_{t+1} O_C\Pi^\progA_t\|$, and also consider $N_{\progB}:=\|\Pi^\progB_t |\phi_t\>\|$ appearing in the definition of the progress measure $\Psi_t$.
As we will see below, when proving the final claim of Lemma~\ref{lem:progEvol}, we need that the scalar in front of $N_{\progB}^2$ in the definition of $\Psi_t$ must be 
strictly larger than the scalar in front of $N_{\progB,\actv}^2$ in (\ref{eq:sansansan}).
While currently those scalars are $3$ and $2$, respectively, we could have taken them to be
$1+\alpha$ and $1+\alpha/2$ for any $\alpha>0$. That is because we could have used in (\ref{eq:sansansan}) the fact that 
\[
(N_{\progB,\actv}+N_{\progC\leftarrow\progA})^2 \le (1+\alpha/2)N_{\progB,\actv}^2 + (1+2/\alpha)N_{\progC\leftarrow\progA}^2.
\]
Note that this would also require the second claim of Lemma~\ref{lem:progEvol} to be changed to $q_{succ}\le \Psi_\tau + \frac{1+1/\alpha}{n-\tau}$.
We have chosen the scalar $3$ instead of $1+\alpha$ in the definition of $\Psi_t$ for the sake of simplicity.
\end{rmk}

\begin{proof}[Proof of the final claim of Lemma~\ref{lem:progEvol}]
For conciseness, let $n_t:=n-t-1$.
Note that $\Pi^\progB_{t+1} O_Q$, $\Pi^\progB_{t+1} O_C$, $\Pi^\progC_{t+1} O_Q$, $\Pi^\progC_{t+1} O_C$ have orthogonal images. Hence, by Lemma~\ref{lem:PiBCOQC} and the fact that $U_{t+1}$ commutes with both $\Pi^\progC_{t+1}$ and $\Pi^\progB_{t+1}$, we have
\begin{align*}
\Psi_{t+1}  =\, &
\|\Pi^\progC_{t+1} U_{t+1}O_p|\phi_t\>\|^2+3\|\Pi^\progB_{t+1} U_{t+1}O_p|\phi_t\>\|^2
\\
=\, &
(1-p)\|\Pi^\progC_{t+1} O_Q|\phi_t\>\|^2
+p\|\Pi^\progC_{t+1} O_C|\phi_t\>\|^2
\\ & \quad + 3(1-p)\|\Pi^\progB_{t+1} O_Q|\phi_t\>\|^2
+3p\|\Pi^\progB_{t+1} O_C|\phi_t\>\|^2
\\ \le \, &
 \|\Pi^\progC_t |\phi_t\>\|^2 
+ 2p\|\Pi^{\progB,\actv}_t|\phi_t\>\|^2
+ \frac{2p}{n_t}
\\ & \quad +
3(1-p)\Big(\| \Pi^{\progB,\actv}_t|\phi_t\> \| + \frac{2}{\sqrt{n_t}}\Big)^2
  + 3\|\Pi^{\progB,\pasv}_t|\phi_t\>\|^2.
  \end{align*}
Note that
  \[
   \|\Pi_t^\progC |\phi_t\>\|^2 + 3\|\Pi_t^{\progB,\pasv}|\phi_t\>\|^2 = \Psi_t - 3\|\Pi_t^{\progB,\actv}|\phi_t\>\|^2,
  \]
  therefore we have
  \begin{align*}
\Psi_{t-1} - \Psi_t
\le\, &
\frac{2p}{n_t} -(3- 2p)\|\Pi^{\progB,\actv}_t|\phi_t\>\|^2
+ 3(1-p)\Big(\| \Pi^{\progB,\actv}_t|\phi_t\> \| + \frac{2}{\sqrt{n_t}}\Big)^2
 \\ = \, &
\frac{48 - 84 p + 38 p^2 }{pn_t}
- p\bigg(\|\Pi_t^{\progB,\actv}|\phi_t\>\| 
- \frac{6(1-p)}{p\sqrt{n_t}}\bigg)^2
\\ \le \, &
\frac{48}{pn_t}.
\end{align*}
(In above, we used that $-(3-2p)+3(1-p)$ is strictly negative, which, in Remark~\ref{rmk:FourInProg}, concerns the scaling of $N_{\progB,\actv}^2$ and $N_{\progB}^2$.)
\end{proof}

\section{Almost Optimal Algorithm}
\label{sec:algoSketch}

The main contribution of this work is showing an $\Omega(\max\{np,\sqrt{n}\})$ query lower bound for search with a noisy oracle of rate $p$. We get the bound by combining Theorem~\ref{thm:main} with the tight bound in the noiseless setting~\cite{bennett:searchLowerB}. In this section we show that this lower bound is asymptotically optimal for the dephasing noise and almost asymptotically optimal for the depolarizing noise.
In particular, we present an 
$\OO(\max\{np(1+p\log n),\sqrt{n}\})$
quantum noisy-query algorithm when $1-p=\Omega(1)$, and the $p\log n$ term and the bound on $1-p$ can be omitted when the dephasing noise is considered.

We do not claim the novelty of this algorithm, as it is a combination of well studied techniques, and for that reason, we present the algorithm and its analysis only informally. The correctness of the algorithm and its claimed running time can be established following the now-somewhat-standard analysis of Grover's algorithm.

\bigskip

Because no algorithm can succeed when no element is marked, we assume that at least one marked element exists. If the number of marked elements in non-constant, say, $n^{1/3}$, we can likely achieve further speedups, as it is done in the case of noiseless search. In this work, however, we do not investigate the dependence of the query complexity on the number of marked elements.

For all upper bounds, we assume that the noisy oracle does not signal occurrences of errors via flag bits. If it did, the algorithm presented here would still work by simply ignoring those bits.
The algorithm that we present works even if the noise rate $p$ is not known in advance.%
\footnote{To allow the running time to depend on an unknown $p$, we slightly deviate from the computational model described in Section~\ref{sec:model}, because that model does not allow the computation to terminate early.}

\subsection{Classical checking subroutine}
\label{sec:MajCheck}

The algorithm that we are about to present requires a classical subroutine that with high precision can check if a given element $x$ is marked, that is, compute $f(x)$. Here, by ``classical'' we mean that we never call this subroutine in a superposition over inputs.

For the dephasing noise, we simply prepare the state $|x\>_\regQi|0\>_\regQo$, call the noisy oracle on this state, and then measure the register $\regQo$. Even if the noise has caused dephasing, we are still guaranteed that our measurement yields $f(x)$. Thus, using a single query, with probability $1$ we output $f(x)$.

For the depolarizing noise, we use a similar procedure, but repeat it $\Theta(\log n)$ times. More precisely, suppose $p=1-r$ for some $r=\Omega(1)$. If we prepare the state $|x\>_\regQi|0\>_\regQo$, call the noisy oracle on this state, and then measure the register $\regQo$, then with probability $1/2+r/2$ we will obtain $f(x)$ and with probability $1/2-r/2$ we will obtain $\neg f(x)$. Repeating this procedure $\Theta(\log n)$ times and taking the majority vote, we can ensure that the correct value $f(x)$ is output with probability at least $1-o(1/(n\log n))$.

Below, we will simply refer to this subroutine as the ``checking procedure'' and denote its cost as $C$. We note that the algorithm below will work correctly for any noise model that independently affects each oracle call with probability $p$ that is bounded below $1/2$ by at least a constant.

\subsection{Truncated run of Grover's algorithm}

For now, let us assume that the noise rate $p\ge 1/\sqrt n$ is known; we will deal with the more general scenario at the very end.
Let $k$ denote the number of marked elements.
In this section, we sketch a construction of an algorithmic subroutine that achieves the following task.

\begin{clm}
\label{clm:TruncGrover}
There exists an algorithm that, given $\OO(1/p+C)$ calls to a noisy oracle, either outputs $\mathsf{fail}$ or, with probability at least $\Omega(1/(np^2))$, it outputs $x\in[n]$. Moreover, if the algorithm outputs $x\in[n]$, then $x$ is marked with probability at least $1-o(1/n)$. The algorithm works regardless of the number of marked elements.
\end{clm}

If we run Grover's algorithm for $t=1/p$ iterations, with at least a constant probability no oracular error occurs in any of the iterations. Assuming there were indeed no errors, if we then measure the resulting state in the computational basis, the probability that we get a marked element is 
\[
\Big(\sin\big((2t+1)\arcsin \sqrt{k/n}\big)\Big)^2,
\]
which is $\Omega(t^2/n)=\Omega(1/(np^2))$ if $k=\OO(n/t^2)=\OO(np^2)$ for a sufficiently small constant behind the big-$\OO$ notation.

To deal with the scenario when $k=\Omega(np^2)$, we simply run instances of Grover's algorithm for $1,2,4,8,\ldots, 1/p$ iterations. If $k=\Omega(np^2)$, then $1/p=\Omega(\sqrt{n/k})$, and with at least a constant probability at least one of these runs will return a marked element.

Finally, for all elements returned by the above executions of Grover's algorithm, we use the checking procedure from Section~\ref{sec:MajCheck}. If at least for one element $x$ the checking procedure reports $f(x)=1$, we return that $x$; otherwise we return $\mathsf{fail}$.

\subsection{Parallel search for a marked element}

To find a marked element with a probability close to $1$, we employ a parallel repetition of the subroutine of Claim~\ref{clm:TruncGrover}. In particular, for every constant $\epsilon>0$, we can choose a sufficiently large constant $c_\epsilon$ so that, by repeating the subroutine  $c_\epsilon np^2$ times, the probability that we return $\mathsf{fail}$ for all of them is less than $\epsilon/2$. The probability that we return $x$ having $f(x)=0$ is $o(1)$, and thus also less than $\epsilon/2$.

This results in the total number of oracle calls being $\OO(np^2(1/p+C))$, which is $\OO(np)$ for the dephasing noise and $\OO(np(1+p\log n))$ for the depolarizing noise.

Finally, to deal with the scenario when $p$ is not known in advance, first we simply run Grover's algorithm and check the result using the checking procedure from Section~\ref{sec:MajCheck}. If a marked element has not been found this way, we guess the value of $p$ by starting with $p=1/\sqrt{n}$ and then doubling the value of $p$ for every subsequent guess, until we succeed (or reach $p=1$). For a given guess of $p$, we run the algorithm above and, if that algorithm outputs $x$, we also output $x$ and terminate the computation.

\section{Discussion}
\label{sec:discussion}

\paragraph{Hybrid, partially-noisy model.}

In this work, we have considered a scenario where the algorithm has  access to only a single type of oracle, namely, a noisy quantum oracle. However, similarly to~\cite{Rosmanis:2022:hybrid,hamoudi:2022:tradeoffs}, we might also consider a hybrid scenario where the algorithm has access to both a noisy quantum oracle and a noiseless classical oracle. Because of the similarities between the proof techniques used in this paper and those used by Hamoudi, Liu, and Sinha in~\cite{hamoudi:2022:tradeoffs}, we suspect that the present techniques could be used to provide tight lower bounds in this hybrid, partially-noisy setting.

\paragraph{Proof without flag bits.}

The completely depolarizing channel can be written as a composition of two non-destructive projective measurements on two mutually unbiased basis. In particular, it can be written as a dephasing on the computational basis followed by a dephasing on the Fourier basis.

Using this observation, one can likely reprove the results in this paper without using the error-signaling flag bits. In particular, we could use the record register to purify both the dephasing on the computational basis and the dephasing on the Fourier basis, but only use the part of the record corresponding to the former when defining the progress measure.

\section*{Acknowledgements}

The author would like to thank Yassine Hamoudi, Atsuya Hasegawa, Fran\c{c}ois Le Gall, Han-Hsuan Lin, and Qisheng Wang for fruitful and insightful discussions.
The author was supported by JSPS KAKENHI Grant No.~JP20H05966 and MEXT Quantum Leap Flagship Program (MEXT Q-LEAP) Grant No.~JPMXS0120319794.
Part of the writing of this work was done while visiting Academia Sinica, and the author is grateful for their hospitality.

\appendix

\section{Noisy Search for Random Functions}
\label{app:randFunc}

Here we consider the unstructured search problem in the random function setting. The problem is, given an oracle access to a uniformly random function $f\colon[n]\rightarrow[m]$, to find $x$ that evaluates to some given value $\marked\in[m]$. We can think of it as the problem of inverting a one-way function. 

As before, we consider the scenario where every oracle call independently suffers from the same noise. Similarly to Theorem~\ref{thm:main}, which concerns the worst case setting, we have the following similar result in the random function setting.

\begin{thm}
\label{thm:mainRand}
For both the depolarizing and the dephasing noise of rate $p>0$, every algorithm that solves the unstructured search problem in the random function setting with probability at least $1-\epsilon$ must make at least $mp(1-\epsilon)/32-1$ queries to the noisy oracle.
This lower bound holds even if for every oracle call the algorithm is provided with a bit indicating whether the oracle call was performed properly.
\end{thm}

The proof of Theorem~\ref{thm:mainRand} follows closely to that of Theorem~\ref{thm:main}. Therefore, for sake of brevity and readability, we do not present here the proof of Theorem~\ref{thm:mainRand} in full detail, but instead we just sketch it, mostly focussing on differences between the two proofs.

\subsection{Truth and record registers}
\label{sec:TruthRecord}

As before, we extend the computation by introducing the truth and the record registers, which purify the random choice of the function and the non-unitary action of the noise, respectively. Their joint state is used to quantify the progress of the computation.

\paragraph{Truth register.}

The truth register $\regFunc$ now corresponds to the set of all functions $f\colon[n]\rightarrow[m]$, and thus its corresponding space is of dimension $m^n$. We decompose it as a concatenation of $n$ subregisters $\regFunc_x$, where $x\in[n]$ and $\regFunc_x$ holds the value of $f(x)$. That is, a function $f$ is stored in the truth register as $|f\>_\regFunc = \bigotimes_{x\in[n]} |f(x)\>_{\regFunc_x}$. 
This register controls the action of the faultless oracle, and it is initialized as the uniform superposition over all functions.

Let $|\unifR\>$ be the uniform superposition over all outputs, namely, $|\unifR\>:=\sum_{y\in[m]}|y\>/\sqrt{m}$. Hence, the initial state of the truth register can be written as $\bigotimes_{x\in[n]} |\unifR\>_{\regFunc_x}$.

As before, we assume that the query-input register is also used by an algorithm to return the answer. If the algorithm answers $x$, then we want that $f(x)=\marked$, while we do not care about the values of $f$ on other inputs. Therefore, the projector on the successful outcomes is
\[
\Pi_{succ} := \sum_{x\in[n]}
\Big(|\marked\>\<\marked|_{\regFunc_x}\otimes \bigotimes_{x'\in[n]\setminus\{x\}} I_{\regFunc_{x'}}\otimes|x\>\<x|_{\regQi}\Big).
\]

Because of the tensor product structure of the space corresponding to the register $\regFunc$, we can now write the \emph{extended faultless oracle} call as
\begin{align*}
O := & \sum_{f}|f\>\<f|_\regFunc\otimes\sum_{x\in[n]}\sum_{y\in[m]}|x,y\oplus f(x)\>\<x,y|_\regQ \\
= & \sum_{x\in[n]}\sum_{y,y'\in[m]}
\Big(|y'\>\<y'|_{\regFunc_x}\otimes \bigotimes_{x'\in[n]\setminus\{x\}} I_{\regFunc_{x'}}
\Big)
\otimes |x,y\oplus y'\>\<x,y|_\regQ.
\end{align*}

\paragraph{Record register.}

As before, the proof actually deals with the dephasing noise which with probability $p$ non-destructively measures $\regQi$ in the computational basis and provides the algorithm with a flag bit indicating whether this measurement has taken place. Therefore, the record register and its expansion by oracle calls is defined exactly as before. Namely, we define
\[
O_Q := \big(I_{\regQi}\otimes |\bot\>_{\regR_t}\otimes |\flagQ\>_{\regW^+}\big) O
\qAnd
O_C := \Big(\sum_{x\in[n]}|x\>\<x|_\regQi\otimes|x\>_{\regR_t}\otimes|\flagC\>_{\regW^+}\Big)O.
\]
However, now the way how we use the record register to define spaces $\cH^{\progA},\cH^{\progB},\cH^{\progC}$ will be different.

\subsection{Progress measure}

Again we quantify the progress of the computation by considering the state $|\phi_t\>$ of the overall systems just before $(t+1)$-th oracle call, and inspecting its overlap on certain subspaces $\cH^{\progB}$ and $\cH^{\progC}$.

\paragraph{Projectors $\Pi^{\progA}_{t},\Pi^{\progB,\actv}_{t},\Pi^{\progB,\pasv}_{t},\Pi^{\progC}_{t}$.}

The projector $\Pi^{\progC}_{t}$ is defined essentially the same way as in the worst case setting, addressed in the main text, while the definitions of  $\Pi^{\progA}_{t}$ and  $\Pi^{\progB}_{t}$ are inspired by \cite{Rosmanis2021}, where noiseless search was analyzed. All three of them act on registers $\regFunc\regR$ (and, implicitly, as the identity on other registers).

As before, let $\Pi^{\progC}_{t}$ project on all those function-record pairs $(f,R)$ such that there exists $x\in\inrec{R}$ such that $f(x)$ is the searched element $\marked$. Unlike before, such an $x$ might not be unique.
Towards defining $\Pi^{\progA}_{t}$ and  $\Pi^{\progB}_{t}$, first define an ``approximation'' of $|\marked\>$ as 
\[
|\tilde\marked\> := 
 \sqrt{\frac{m-1}{m}}|\marked\> - \frac1{\sqrt{m(m-1)}}\sum_{y\in[m]\setminus\{\marked\}}|y\>,
\]
which is orthogonal to the uniform superposition $|\unifR\>$. The state $|\tilde\marked\>$ serves a somewhat similar role to the approximation $|\tilde{f}_{z,R}\>$ of $|f_z\>$ defined in Section~\ref{sec:ActAndPas}, except that $|\tilde\marked\>$  is independent of the record $R$. 
 We define
\begin{align*}
\Pi^{\progC}_{t} := \, & \sum_{R\in\rec^t} \sum_{\substack{f\\\marked\in f(\inrec{R})}}|f\>\<f|_{\regFunc} \otimes |R\>\<R|
\\ = \, & \sum_{R\in\rec^t} \Big[ I_{\regFunc_{\inrec{R}}}-\bigotimes_{x\in\inrec{R}}(I_m-|\marked\>\<\marked|)_{\regFunc_x}\Big]
\otimes I_{\regFunc_{[n]\setminus\inrec{R}}} \otimes |R\>\<R|,
\\ \Pi^{\progB}_{t} := \, &
\sum_{R\in\rec^t} 
\bigotimes_{x\in\inrec{R}}(I_m-|\marked\>\<\marked|)_{\regFunc_x}
\otimes
\Big[ I_{\regFunc_{[n]\setminus\inrec{R}}}-\bigotimes_{x\in[m]\setminus\inrec{R}}(I_m-|\widetilde\marked\>\<\widetilde\marked|)_{\regFunc_x}\Big]
\otimes |R\>\<R|,
\\ \Pi^{\progA}_{t} := \, &
\sum_{R\in\rec^t} 
\bigotimes_{x\in\inrec{R}}(I_m-|\marked\>\<\marked|)_{\regFunc_x}
\otimes
\bigotimes_{x\in[m]\setminus\inrec{R}}(I_m-|\widetilde\marked\>\<\widetilde\marked|)_{\regFunc_x}
\otimes |R\>\<R|,
\end{align*}
where we assume that $I_{\regFunc_{\inrec{R}}}=\bigotimes_{x\in\inrec{R}}(I_m-|\marked\>\<\marked|)_{\regFunc_x}=1$ when $\inrec{R}=\emptyset$.
Here, for a subset $S\subseteq [n]$, $\regFunc_S$ is the concatenation of all registers $\regFunc_x$ with $x\in S$.

We further decompose $\Pi^{\progB}_{t} \otimes I_\regQi $ as the sum of two projectors $\Pi^{\progB,\actv}_{t}$ and $\Pi^{\progB,\pasv}_{t}$, where
\begin{multline*}
\Pi^{\progB,\actv}_{t} := 
\sum_{R\in\rec^t} 
\bigotimes_{x'\in\inrec{R}}(I_n-|\marked\>\<\marked|)_{\regFunc_{x'}}
\otimes\sum_{x\in[m]\setminus\inrec{R}} \bigg(|\widetilde\marked\>\<\widetilde\marked|_{\regFunc_{x}}
\\ 
\otimes \bigotimes_{x''\in[m]\setminus(\inrec{R}\cup\{x\})}(I_n-|\widetilde\marked\>\<\widetilde\marked|)_{\regFunc_{x''}}
\bigg)
\otimes |R\>\<R| \otimes |x\>\<x|.
\end{multline*}
The definition of $\Pi^{\progB,\actv}_{t}$ and its corresponding space $\cH^{\progB,\actv}_{t}$ are chosen to indicate the case when the quantum algorithm has found one and \emph{only one} marked element. This limitation to one marked element, first of all, is what we will care about when we are considering transitions from $\cH^\progA$ to $\cH^{\progB,\actv}$. And, second, this limitation will ensure that neither $O_Q$ nor $O_C$ can cause transitions from $\cH^{\progB,\actv}$ to $\cH^{\progB,\pasv}$.

\paragraph{Progress measure and the main lemma.}

Similarly as before, we define the progress measure as 
\[
\Psi_t:=\|\Pi^\progC_{t}|\phi_t\>\|^2 + 4 \|\Pi^\progB_{t}|\phi_t\>\|^2.
\]
Note that now the constant in front of $\|\Pi^\progB_{t}|\phi_t\>\|^2$ is $4$, not $3$. As per Remark~\ref{rmk:FourInProg}, we could reduce this constant, but we leave it to be $4$ for the sake of simplicity. The reason why it is larger now than in the worst case setting considered in the main text is that now we will not claim that the space $\cH^{\progB,\pasv}$ is invariant under oracle calls. Indeed, because $\cH^{\progB,\pasv}$ contains some states that essentially correspond to having found multiple marked elements, we have $\Pi^{\progC}O_C\Pi^{\progB,\pasv}\ne 0$.%
\footnote{Regarding the quantum oracle call, from Claims~\ref{clm:nonAlterRAN}, \ref{clm:OQonActRAN}, and~\ref{clm:OQonPiARAN} below, it can be seen that we still have $O_Q\Pi^{\progB,\pasv}_t= \Pi^{\progB,\pasv}_{t+1}O_Q$.}

Similarly to Lemma~\ref{lem:progEvol} for the worst case setting, we will now show the following.
\begin{lem}
\label{lem:progEvolRand}
We have
\begin{subequations}
\begin{align}
& \Psi_0 = 0, \\
& q_{succ} \le \Psi_\tau +  \frac2{m}, \\
& \Psi_{t+1}-\Psi_{t} \le \frac{32}{p(m-1)}.
\end{align}
\end{subequations}
\end{lem}

Given Lemma~\ref{lem:progEvolRand}, we can see that the final success probability is upper bounded as
\[
q_{succ} \le \frac{2}{m-1} + \frac{32\tau}{p(m-1)}.
\]
So, if we want $q_{succ}$ to be at least $1-\epsilon$, we must have
\[
\tau \ge \frac{p(1-\epsilon)(m-1)-2p}{32}
\ge \frac{pm(1-\epsilon)}{32}-1,
\]
which proves Theorem~\ref{thm:mainRand}. It is left to prove Lemma~\ref{lem:progEvolRand} itself. As before, its first two claims are easier to show.

\paragraph{First two claims of Lemma~\ref{lem:progEvolRand}.}

At the beginning of the computation, the record is empty and the truth register is in the state $|\unifR\>^{\otimes n}$. The first claim of Lemma~\ref{lem:progEvolRand} then holds because $\<\widetilde\marked|\unifR\>=0$.

As for the second claim, recall the projector on successful outcomes $\Pi_{succ}$, given in Section~\ref{sec:TruthRecord}.
Note that both $\Pi_{succ}$ and $\Pi^\progC_\tau$ are orthogonal in the computational basis, so they commute, and thus $\Pi_{succ}$ also commutes with $\Pi^\progA_\tau+\Pi^\progB_\tau$.
So, using the same chain of inequalities as in the proof of the second claim of Lemma~\ref{lem:progEvol}, we have
$q_{succ} \le \Psi_\tau + 2\|\Pi_{succ}\Pi^\progA_\tau\|^2$.
Finally, by inspecting the expressions for $\Pi_{succ}$ and $\Pi^\progA_\tau$ above, we can see that the non-zero parts of $\Pi_{succ}\Pi^\progA_\tau$ correspond to $x\in[n]\setminus\inrec{R}$, and thus
\[
\big\|\Pi_{succ}\Pi^\progA_\tau\big\| = 
\big\||\marked\>\<\marked|(I_n-|\tilde\marked\>\<\tilde\marked|)_{\regFunc_x}\big\| 
= 1/\sqrt{m}.
\]

\subsection{Transitions among progress-defining subspaces}

To prove the third claim of Lemma~\ref{lem:progEvolRand}, here we establish equivalents of Claims~\ref{clm:nonAlter}, \ref{clm:OQonAct}, \ref{clm:OQonPiA}, \ref{clm:OConPiA}. Unlike before, now the space $\cH^{\progB,\pasv}$ is not invariant under oracle calls, so we do not have an equivalent of Claim~\ref{clm:OQonPas}. However, even without this claim, we can still establish $\Pi^\progB_{t+1} O_Q\Pi^\progA_t = \Pi^{\progB,\actv}_{t+1} O_Q\Pi^\progA_t $, which was the main purpose of Claim~\ref{clm:OQonPas}.

\bigskip

First we claim the following, which is word by word equivalent to Claim~\ref{clm:nonAlter}. We omit its proof, which is essentially the same as that of Claim~\ref{clm:nonAlter}.

\begin{clm}
\label{clm:nonAlterRAN}
We have $O_Q\Pi^{\progC}_t = \Pi^{\progC}_{t+1}O_Q$ and $O_C\Pi^{\progC}_t = \Pi^{\progC}_{t+1}O_C\Pi^{\progC}_t$, and the images of $\Pi^\progC_{t+1} O_C\Pi^\progC_t$ and $\Pi^\progC_{t+1} O_C(\Pi^\progB_t+\Pi^\progA_t)$ are orthogonal.
\end{clm}

As mentioned before, now the equivalent of Claim~\ref{clm:OQonPas} does not hold. The next claim is essentially a union of a weaker version of Claim~\ref{clm:OQonPas} and the equivalent of Claim~\ref{clm:OQonAct}.

\begin{clm}
\label{clm:OQonActRAN}
We have $\Pi^{\progB}_{t+1}O_Q\Pi^{\progB,\actv}_t=\Pi^{\progB,\actv}_{t+1}O_Q\Pi^{\progB,\actv}_t$, $\Pi^{\progB}_{t+1}O_Q\Pi^{\progB,\pasv}_t=\Pi^{\progB,\pasv}_{t+1}O_Q\Pi^{\progB,\pasv}_t$, and
$O_C\Pi^{\progB,\actv}_t=(\Pi^\progC_{t+1}+\Pi^\progA_{t+1})O_C\Pi^{\progB,\actv}_t$.
\end{clm}

\begin{proof}
In this proof, we ignore the workspace, which is irrelevant here. In addition, we will ignore the flag bits $|\flagQ\>_{\regW_+}$ and $|\flagC\>_{\regW_+}$ introduced by $O_Q$ and $O_C$, respectively.

The space $\cH^{\progB,\actv}_t$ is spanned by vectors in form
\begin{equation}
\label{eq:HBactVec}
|\tilde\marked\>_{\regFunc_{x}}
\otimes |rest_{R,x}\>_{\regFunc_{[n]\setminus\{x\}}}
\otimes |R\>_\regR|x\>_\regQi |y\>_\regQo
\end{equation}
where $R\in\rec^t$, $x\in[n]\setminus\inrec{R}$, $y\in[m]$, and 
\[
|rest_{R,x}\>_{\regFunc_{[n]\setminus\{x\}}} := \bigotimes_{x'\in\inrec{R}}|z_{x'}\>_{\regFunc_{x'}} \otimes \bigotimes_{x''\in[n]\setminus(\inrec{R}\cup\{x\})} |\zeta_{x''}\>_{\regFunc_{x''}}
\]
satisfies $z_{x'}\ne\marked$ and $\<\tilde\marked|\zeta_{x''}\>=0$ for all $x'\in\inrec{R}$ and $x''\in[n]\setminus(\inrec{R}\cup\{x\})$.

\paragraph{Claims concerning $O_Q$.}
If we apply $O_Q$ to the vector (\ref{eq:HBactVec}), it will entangle registers $\regFunc_x$ and $\regQo$, append $\bot$ to $R$ (note: $\inrec{(R\app\bot)}=\inrec{R}$), and keep the other registers unchanged.
The part of the resulting state which still has $|\tilde\marked\>$ in the register $\regFunc_x$ belongs to $\cH^{\progB,\actv}_t$, while the part of the resulting state with the register $\regFunc_x$ being orthogonal to $|\tilde\marked\>$ belongs to $\cH^{\progA}_t$. We thus have shown 
$O_Q\Pi^{\progB,\actv}_t=(\Pi^{\progA}_{t+1}+\Pi^{\progB,\actv}_{t+1})O_Q\Pi^{\progB,\actv}_t$. 

Then, the part of the claim concerning $\Pi^{\progB,\pasv}$ holds because, aside from introducing states $|\bot\>$  and $|\flagQ\>$, the quantum oracle call $O_Q$ is its own inverse.

\paragraph{Claim concerning $O_C$.}
If we apply $O_C$ to the vector (\ref{eq:HBactVec}), we get
\begin{align*}
&
\sqrt{\frac{m-1}{m}}
|\marked\>_{\regFunc_{x}}
\otimes |rest_{R,x}\>_{\regFunc_{[n]\setminus\{x\}}}
\otimes |R\app x\>_\regR|x\>_\regQi |y\oplus\marked\>_\regQo
\\
&
+
\frac{1}{\sqrt{m(m-1)}}
\sum_{y'\in[m]\setminus\{\marked\}}
|y'\>_{\regFunc_{x}}
\otimes |rest_{R,x}\>_{\regFunc_{[n]\setminus\{x\}}}
\otimes |R\app x\>_\regR|x\>_\regQi |y\oplus y'\>_\regQo.
\end{align*}
Now that $x$ is in the record, the former state is in $\cH^{\progC}_t$, while the latter is in $\cH^{\progA}_t$.
\end{proof}

Similarly to Claim~\ref{clm:OQonPiA}, we now have the following.

\begin{clm}
\label{clm:OQonPiARAN}
We have $\Pi^\progB_{t+1} O_Q\Pi^\progA_t = \Pi^{\progB,\actv}_{t+1} O_Q\Pi^\progA_t $, and its norm is 
$\frac{\sqrt{2m-3}}{m-1}$.
\end{clm}

\begin{proof}
Informally, if we query $x$ that is already in the record, nothing changes. Hence, the only way to transit from $\cH^{\progA}$ to $\cH^{\progB}$ is by querying $x$ that is not in the record. More formally, we have
\begin{align*}
 \Pi^{\progB}_{t+1}O_Q\Pi^{\progA}_{t} = \, &
\sum_{R\in\rec^t} 
\bigotimes_{x'\in\inrec{R}}(I_m-|\marked\>\<\marked|)_{\regFunc_{x'}}
\\ & \qquad
\otimes
\sum_{x\in[n]\setminus\inrec{R}}\sum_{y,y'\in[m]}
\bigg(
\bigotimes_{x''\in[n]\setminus(\inrec{R}\cup\{x\})}(I_m-|\widetilde\marked\>\<\widetilde\marked|)_{\regFunc_{x''}}
\\ & \qquad\qquad\qquad\qquad\qquad\quad
\otimes \big(|\tilde\marked\>\<\tilde\marked|y'\>\<y'|(I_m-|\tilde\marked\>\<\tilde\marked|)\big)_{\regFunc_x} \bigg)
\\ & \qquad
\otimes |R\app\bot\>\<R|
\otimes |\flagQ\>
\otimes |x\>\<x|
\otimes |y\oplus y'\>\<y|.
\end{align*}
It is easy to see that the image of the above operator is in $\cH^{\progB,\actv}_{t+1}$, and its norm is
\begin{align*}
 \|\Pi^{\progB}_{t+1}O_Q\Pi^{\progA}_{t}\| = \, &\Big\|
\sum_{y,y'\in[m]} \<\tilde\marked|y'\>\<y'|(I_m-|\tilde\marked\>\<\tilde\marked|)
\otimes |y\oplus y'\>\<y|\Big\|
\end{align*}
which is computed exactly in \cite[Proof of Claim~11]{Rosmanis2021} as being $\frac{\sqrt{2m-3}}{m-1}$.
\end{proof}

Similarly to Claim~\ref{clm:OConPiA}, we now have the following.

\begin{clm}
\label{clm:OConPiARAN}
We have $\Pi^\progB_{t+1} O_C\Pi^\progA_t=0$ and $\|\Pi^\progC_{t+1} O_C\Pi^\progA_t\|^2=1/m$.
\end{clm}

\begin{proof}
Using the definitions of $O_C$ and $\Pi^\progA_t$, we write
 \begin{align*}
O_C\Pi^\progA_t = \, & 
\bigg(
\sum_{x\in[n]}\sum_{y,y'\in[m]}
|y\>\<y|_{\regFunc_x}\otimes  I_{\regFunc_{[n]\setminus\{x\}}}
\otimes |x\>\<x|_\regQi
\otimes |y'\oplus y\>\<y'|_\regQo
\otimes|x\>_{\regR_{t+1}}
\bigg)
\\ & \cdot\bigg(
\sum_{R\in\rec^t} 
\bigotimes_{x'\in\inrec{R}}(I_n-|\marked\>\<\marked|)_{\regFunc_{x'}}
\otimes
\bigotimes_{x''\in[n]\setminus\inrec{R}}(I_n-|\widetilde\marked\>\<\widetilde\marked|)_{\regFunc_{x''}}
\otimes |R\>\<R|
\bigg).
\end{align*}
where we have ignored the flag qubit $|\flagC\>_{\regW^+}$ for the sake of brevity.
Let us decompose the sum above into two parts, namely, $O_C\Pi^\progA_t=M'+M''$, where 
$M'$ corresponds to all the terms with $y\ne\marked$ and 
$M''$ corresponds to all the terms with $y=\marked$ and $x\notin\inrec{R}$.
  It is easy to see that all the terms with $y=\marked$ and $x\in\inrec{R}$ are $0$.

One can see that, for states in the image of $M'$, it both holds that the registers $T_{x'}$ with $x'$ in the record are orthogonal to $|\marked\>$ and that the registers $T_{x''}$ with $x''$ outside the record are orthogonal to $|\widetilde\marked\>$. Hence, the image of $M'$ is contained in $\cH^{\progA}_{t+1}$.

Now let us consider $M''$, for which we consider only $y=\marked$. For brevity, let $U^{\oplus\marked}:=\sum_{y'\in[m]}|y'\oplus \marked\>\<y'|$. We have
 \begin{multline*}
M'' =  
\sum_{R\in\rec^t} 
\sum_{x\in [n]\setminus\inrec{R}}
\bigg(
\big(|\marked\>\<\marked | (I_n-|\widetilde\marked\>\<\widetilde\marked|) \big)_{\regFunc_x}
 \otimes \hspace{-10pt}
\bigotimes_{x''\in[n]\setminus(\inrec{R}\cup\{x\})}
\hspace{-10pt}
(I_n-|\widetilde\marked\>\<\widetilde\marked|)_{\regFunc_{x''}} 
\\ 
\otimes
\bigotimes_{x'\in\inrec{R}}(I_n-|\marked\>\<\marked|)_{\regFunc_{x'}} \bigg)
\otimes |R\app x\>\<R|
\otimes |x\>\<x|_\regQi \otimes U^{\oplus\marked}_\regQo.
\end{multline*}
Fix $R\in\rec^t$ and $x\in [n]\setminus\inrec{R}$, and consider the image of the corresponding term in the sum above. The restriction of that image to registers $\regFunc_x\regR$ is in the state $|\marked,R\app x\>$, therefore that image is contained in $\cH^{\progC}_{t+1}$. As the result, the image of $M''$ is also contained in $\cH^{\progC}_{t+1}$.

We have shown that
\[
\Pi^\progA_{t+1} O_C\Pi^\progA_t = M', \qquad
\Pi^\progB_{t+1} O_C\Pi^\progA_t = 0, \qquad
\Pi^\progC_{t+1} O_C\Pi^\progA_t = M'',
\]
and it remains to bound the norm of $M''$.
In the expression for $M''$, the terms corresponding to distinct $x$ and distinct $R$ are orthogonal, therefore one can see that
 $
\| M'' \|
= 
\big\|\<\marked|(I_m-|\tilde\marked\>\<\tilde\marked|)\big\|.
$
Recall the uniform superposition $|\unifR\>=\sum_{y\in[m]}|y\>/\sqrt{m}$. Because $|\marked\>$ is a linear combination of $|\tilde\marked\>$ and $|\unifR\>$, we have
$\<\marked|(I_m-|\tilde\marked\>\<\tilde\marked|) = \<\marked|\unifR\>\<\unifR|=\<\unifR|/\sqrt{m}$.
\end{proof}

\subsection{Conclusion of the proof}

Similarly how Claims~\ref{clm:nonAlter}--\ref{clm:OConPiA} imply Lemma~\ref{lem:PiBCOQC}, now, from Claims~\ref{clm:nonAlterRAN}--\ref{clm:OConPiARAN} we get 
\begin{alignat*}{2}
&
\|\Pi^\progB_{t+1} O_Q |\phi_t\>\|^2
& \le& \bigg(\| \Pi^{\progB,\actv}_t|\phi_t\> \|
+ \frac{\sqrt 2}{\sqrt{m-1}}
\bigg)^2
  + \|\Pi^{\progB,\pasv}_t|\phi_t\>\|^2,
\\ 
&
\|\Pi^\progB_{t+1} O_C|\phi_t\>\|^2
 &=& \|\Pi^{\progB,\pasv}_{t+1} O_C \Pi^{\progB,\pasv}_t|\phi_t\>\|^2
 \\
 &&=& \|\Pi^{\progB,\pasv}_t|\phi_t\>\|^2 - \|\Pi^\progC_{t+1} O_C \Pi^{\progB,\pasv}_t|\phi_t\>\|^2,
\\ 
&
\|\Pi^\progC_{t+1} O_Q |\phi_t\>\|^2  & = & \|\Pi^\progC_t |\phi_t\>\|^2,
\\  &
\|\Pi^\progC_{t+1} O_C|\phi_t\>\|^2  
 &\le & \|\Pi^\progC_t |\phi_t\>\|^2 
+ 3\|\Pi^{\progB,\actv}_t|\phi_t\>\|^2
+ 3\|\Pi^\progC_{t+1} O_C\Pi^{\progB,\pasv}_t|\phi_t\>\|^2
+ \frac3{m-1},
\end{alignat*}
where we have used $\sqrt{2m-3}/(m-1)< \sqrt{2/(m-1)}$ for the first inequality and $3/m<3/(m-1)$
 and $(a+b+c)^2\le 3a^2+3b^2+3c^2$ for the last.
Note that, unlike when establishing Lemma~\ref{lem:PiBCOQC}, here we do not have $\Pi^\progC_{t+1} O_C\Pi^{\progB,\pasv}_t=0$.

Following along the same lines as the proof of the final claim of Lemma~\ref{lem:progEvol}, we get 
\begin{align*}
\Psi_{t+1}   \le \, &
 \|\Pi^\progC_t |\phi_t\>\|^2 
+ 3p\|\Pi^{\progB,\actv}_t|\phi_t\>\|^2
+ 3p\|\Pi^\progC_{t+1} O_C\Pi^{\progB,\pasv}_t|\phi_t\>\|^2
+ \frac{3p}{m-1}
\\ & +
4(1-p)\Big(\| \Pi^{\progB,\actv}_t|\phi_t\> \| + \frac{\sqrt{2}}{\sqrt{m-1}}\Big)^2
  + 4\|\Pi^{\progB,\pasv}_t|\phi_t\>\|^2
  \\ &
  - 4p\|\Pi^\progC_{t+1} O_C \Pi^{\progB,\pasv}_t|\phi_t\>\|^2.
  \end{align*}
 The overall scalar of the term $\|\Pi^\progC_{t+1} O_C \Pi^{\progB,\pasv}_t|\phi_t\>\|^2$ is negative, and thus this term can be dropped.
Hence, we have
  \begin{align*}
\Psi_{t-1} - \Psi_t
\le\, &
\frac{3p}{m-1} -(4- 3p)\|\Pi^{\progB,\actv}_t|\phi_t\>\|^2
+ 4(1-p)\Big(\| \Pi^{\progB,\actv}_t|\phi_t\> \| + \frac{\sqrt{2}}{\sqrt{m-1}}\Big)^2
  \\ = \, &
\frac{32-56p+27p^2}{p(m-1)}
- p\bigg(\|\Pi_t^{\progB,\actv}|\phi_t\>\| 
- \frac{4\sqrt{2}(1-p)}{p\sqrt{m-1}}\bigg)^2
\\ \le \, &
\frac{32}{p(m-1)}.
\end{align*}

\end{document}